\documentclass[runningheads]{llncs}

\usepackage{booktabs}   
\usepackage{subcaption} 
\usepackage{wrapfig}
\usepackage{times}
\usepackage{graphicx}

\usepackage{amsmath,amsfonts,amssymb}
\usepackage{textcomp}
\usepackage{xcolor}

\usepackage{xspace}
\usepackage{algorithm}
\usepackage[noend]{algpseudocode}
\usepackage{multirow}
\usepackage{makecell}
\usepackage{caption}
\usepackage{subcaption}
\usepackage{mathtools}
\usepackage{tikz}
\usepackage{url}
\usetikzlibrary{arrows,calc}
\usepackage{enumitem}
\usepackage{pifont}
\definecolor{red2}{rgb}{0.7,0.3,0.2}
\definecolor{green2}{rgb}{0.1,0.7,0.3}
\definecolor{blue2}{rgb}{0.2,0.4,0.7}

\definecolor{eyecancerpink}{rgb}{1.0, 0, 1.0}

\newcommand{\goodmark}{\textcolor{blue2}{\ding{51}}}%
\newcommand{\badmark}{\textcolor{red2}{\ding{55}}}%

\newcommand{\hide}[1]{}

\newcommand{\bbfZ}{\mathbb{Z}}

\newcommand{\bbfS}{\mathbb{S}}
\newcommand{\bbfD}{\mathbb{D}}
\newcommand{\bbfRU}{{[0, 1]}}
\newcommand{\calM}{\mathcal{M}}
\newcommand{\od}{\overline{\mathbf{d}}}
\newcommand{\oD}{\overline{\mathbf{D}}}
\newcommand{\calX}{\mathcal{X}}
\newcommand{\calG}{\mathcal{G}}

\newcommand{\Obs}{\Omega}
\newcommand{\obs}{\omega}

\newcommand{\blank}{\llcorner\!\!\lrcorner}

\newcommand{\zpython}{\textsc{Z3Python}\xspace}
\algnewcommand\algorithmicmatch{\textbf{match}}
\algnewcommand\algorithmicwith{\textbf{with}}
\algnewcommand\algorithmiccase{\textbf{case}}
\algdef{SE}[MATCH]{Match}{EndMatch}[1]{\algorithmicmatch\ {#1}\ \algorithmicwith}{\algorithmicend\ \algorithmicmatch}%
\algdef{SE}[CASE]{Case}{EndCase}[1]{\algorithmiccase\ {#1}:}{\algorithmicend\ \algorithmiccase}%
\algtext*{EndMatch}%
\algtext*{EndCase}%
\algnewcommand{\IfThenElse}[3]{
  \State \algorithmicif\ #1\ \algorithmicthen\ #2\ \algorithmicelse\ #3}
\newcommand{\lCase}[2]{\State \textbf{case} {#1}{:} {#2}}

\newcommand{\ldp}[1]{{\color{green}DP: #1}}

\begin{document}

\title{Verifying Pufferfish Privacy in Hidden Markov Models}

\author{
Depeng Liu\inst{1,2} \and
Bow-Yaw Wang \inst{3} \and
Lijun Zhang\inst{1,2}}

\authorrunning{D. Liu et al.}

\institute{Institute of Software, Chinese Academy of Sciences, Beijing, China
\email{\{liudp,zhanglj\}@ios.ac.cn}\\ 
\and
University of Chinese Academy of Sciences, Beijing, China \and
Institute of Information Science, Academia Sinica, Taipei, Taiwan
\email{bywang@iis.sinica.edu.tw}\\
}

\maketitle

\begin{abstract}
Pufferfish is a Bayesian privacy framework for designing and analyzing
privacy mechanisms. It refines differential privacy, the current gold standard
in data privacy, by allowing explicit prior knowledge in privacy analysis.
In practice, privacy mechanisms often need be
modified or adjusted to specific applications. Their privacy risks have
to be re-evaluated for different circumstances.
Privacy proofs can thus be
complicated and prone to errors. Such tedious tasks are burdensome
to average data curators. In this paper, we propose an automatic
verification technique for Pufferfish privacy. We use hidden Markov
models to specify and analyze discrete mechanisms in Pufferfish privacy. We show that the
Pufferfish verification problem in hidden Markov models is NP-hard.
Using Satisfiability Modulo Theories solvers, we propose an
algorithm to verify privacy requirements. We implement our algorithm
in a prototypical tool called FAIER, and analyze several classic privacy mechanisms in Pufferfish privacy.
Surprisingly, our analysis show that na\"ive discretization of well-established privacy
mechanisms often fails, witnessed by counterexamples generated by FAIER. In discrete \emph{Above Threshold},
we show that it results in absolutely
no privacy. Finally,
we compare our approach with state-of-the-art tools for differential privacy, and show that our verification technique can be efficiently combined with these tools for the purpose of certifying counterexamples and finding a more precise lower bound for the privacy budget $\epsilon$.
\end{abstract}


\section{Introduction}
\label{section:introduction}





Differential privacy is a framework for designing and analyzing privacy
measures~\cite{DR:14:AFDP,D:06:DP}. In the framework, data publishing
mechanisms are formalized as randomized algorithms. On any input data
set, such mechanisms return randomized answers to queries.
In order to preserve privacy, differential privacy
aims to ensure that similar output distributions are
yielded on similar input data sets.
Differential privacy moreover allows data curators
to evaluate privacy and utility quantitatively. The framework
has attracted lots of attention from academia and industry such as Microsoft~\cite{DKY:17:CTDP} and Apple~\cite{apple}.


Pufferfish is a more recent privacy framework which refines
differential privacy~\cite{KM:14:PFMPD}. In differential privacy,
there is no explicit correlation among entries in data sets during
privacy analysis.
The no free lunch theorem~\cite{KM:11:NFLDP}
in data privacy shows that prior knowledge about data sets is
crucial to privacy analysis. The Pufferfish privacy framework hence allows
data curators to analyze privacy with prior knowledge about data
sets. Under the Bayesian privacy framework, it is shown that
differential privacy preserves the same level of privacy if there is
no correlation among entries in data sets.

\vspace{-2pt}
For differential and Pufferfish privacy, data publishing
mechanisms are analyzed --often on paper-- with sophisticated mathematical tools.
The complexity of the problem is high~\cite{GNP:20:CVLPDP}, and moreover,
it is well-known that such proofs are very subtle and error-prone.
For instance, several published variations of differentially private
mechanisms are shown to violate privacy~\cite{CM:15:ppvsvt,LSL:17:usvtdp}. In order
to minimize proof errors and misinterpretation, the
formal method community has also started to develop techniques for
checking differentially private
mechanisms, such as verification techniques based on approximate couplings~\cite{AH:17:SCPDP,BGGHS:16:PDPPC,BGHJP:16:PLTDP,BKO:12:prrdp,BKO:12a:prrdp,FCG:21:CRSEDP}, 
randomness alignments~\cite{WDKZ:20:CDP,WDW:19:shadowdp,ZK:17:lightdp}, model checking~\cite{LWZ:18:MCDPP}
as well as those with well-defined programming semantics~\cite{BCJSV:20:ddppfio,MM:19:PPADP}
and techniques based on testing and searching~\cite{BTDPM:18:dpfinder,BSBV:21:dpsniper,DWWZK:18:DVDP,ZRH:20:DPCheck}.
\vspace{-2pt}

Reality nevertheless can be more complicated than mathematical
proofs.
Existing privacy mechanisms hardly fit their
data publishing requirements perfectly.
These algorithms may be implemented differently when used
in practice.
Majority of differentially private mechanisms utilize
continuous perturbations by applying the Laplace mechanism.
Computing devices however only approximate continuous noises through floating-point
computation, which is discrete in nature. Care must be taken lest privacy should be lost during such
finite approximations~\cite{M:12:SLSBDP}. Moreover, adding continuous noises
may yield uninterpretable outputs for categorical or discrete
numerical data. Discrete noises are hence necessary for such data.
A challenging task for data curators
is to guarantee that the implementation (discrete in nature) meets the specification (often continuous distributions are used).
It is often time consuming -- if not impossible, to carry out privacy analysis for each modification. Automated verification and testing techniques are in this case a promising
methodology for preserving privacy.




\vspace{-3pt}
In this work, we take a different approach to solve the problems
above. We focus on Pufferfish privacy, and propose
a lightweight but automatic verification technique.
We propose a formal model for data publishing mechanisms and
reduce Pufferfish privacy into a verification problem for hidden
Markov models (HMMs).
Through our formalization, data curators can verify their
specialized privacy mechanisms without going through tedious
mathematical proofs.

\vspace{-3pt}
We have implemented our algorithm in a prototypical tool called FAIER (the pufferFish privAcy verifIER).
We consider privacy mechanisms for bounded discrete
numerical queries such as counting. For those queries, classical
continuous perturbations may give unusable answers or even lose
privacy~\cite{M:12:SLSBDP}. We hence discretize privacy mechanisms by
applying discrete perturbations on such queries.
We report case studies derived from
differentially private mechanisms. Our studies show that na\"ive
discretization may induce significant privacy risks. For the \emph{Above Threshold} example,
we show that discretization does not have any privacy at
all. For this example, our tool generates \emph{counterexamples} for an arbitrary small privacy budget $\epsilon$.
Another interesting problem for differential privacy is to find the largest lower bound of $\epsilon$,
below which the mechanism will not be differentially private.
We discuss how our verification approach can be efficiently combined with testing techniques to solve this problem.

\hide{
In~\cite{LWZ:18:MCDPP}, the authors propose Markov chains and Markov
decision processes to model data publishing mechanisms for
differential privacy. Several known
mechanisms are formalized as different Markov models and checked to
satisfy differential privacy.
}


\vspace{-2pt}
Below we summarize the main contributions of our paper:
\vspace{-2pt}
\begin{enumerate}[leftmargin=*]
\item We propose a verification framework for Pufferfish privacy by specifying privacy mechanisms as HMMs and analyzing privacy requirements in the models (Section~\ref{section:hmm}). To our best knowledge, the work of Pufferfish privacy verification had not been investigated before.
\item Then we study the Pufferfish privacy verification problem
on HMMs and prove the verification problem
to be NP-hard (Section~\ref{subsection:complexity}).
\item On the practical side, nevertheless, using SMT solvers,
we design a verification algorithm which automatically verifies Pufferfish privacy (Section~\ref{subsection:checking-pufferfish}).
\item The verification algorithm is implemented into the tool FAIER (Section~\ref{subsection:evaluation}). We then perform case studies of classic mechanisms, such as Noisy Max and Above Threshold.
Using our tool, we are able to catch privacy breaches of the specialized mechanisms (Section~\ref{subsection:noisymax} ~\ref{subsection:above-threshold}).
\item Compared with the state-of-the-art tools DP-Sniper~\cite{BSBV:21:dpsniper} and StatDP~\cite{DWWZK:18:DVDP} on finding the privacy
budget $\epsilon$ (or finding privacy violations) for differential privacy, our tool has advantageous performances in obtaining the most precise results within acceptable time for discrete mechanisms. We propose to exploit each advantage to the full to efficiently obtain a precise lower bound for the privacy budget $\epsilon$ (Section~\ref{section:experiments}).
\end{enumerate}

\section{Preliminaries}
\label{section:preliminaries}

A \emph{Markov Chain} $K = (S, p)$ consists of a finite set $S$ of
\emph{states} and a \emph{transition distribution} $p : S \times S
\rightarrow \bbfRU$ such that $\sum_{t \in S} p (s, t) = 1$ for every
$s \in S$.
A \emph{Hidden Markov Model} (HMM) $H = (K, \Obs, o)$ is a Markov chain $K =
(S, p)$ with a finite set $\Obs$ of \emph{observations} and an
\emph{observation distribution} $o : S \times \Obs \rightarrow \bbfRU$
such that $\sum_{\obs \in \Obs} o (s, \obs) = 1$ for every $s \in
S$. Intuitively,
the states of HMMs are not observable. External observers do not know
the current state of an HMM. Instead, they have a state distribution
(called \emph{information state})
$\pi : S \rightarrow [0, 1]$ with $\sum_{s \in S} \pi (s) = 1$
to represent the likelihood of each state in an HMM.

Let $H = ((S, p), \Obs, o)$ be an HMM and $\pi$ an initial state
distribution. The HMM $H$ can be seen as a (randomized) generator for
sequences of observations. The following procedure generates
observation sequences of an arbitrary length:

\vspace{-7pt}
\begin{enumerate}
\item $t \leftarrow 0$.
\item Choose an initial state $s_0 \in S$ by the initial state
  distribution $\pi$.
\item \label{hmm:repeat}
  Choose an observation $\omega_t$ by the observation distribution
  $o (s_t, \bullet)$.
\item Choose a next state $s_{t+1}$ by the transition distribution
  $p (s_t, \bullet)$.
\item $t \leftarrow t + 1$ and go to \ref{hmm:repeat}.
\end{enumerate}
\vspace{-7pt}

Given an observation sequence $\overline{\omega} =
\omega_0\omega_1\cdots\omega_k$ and a state sequence $\overline{s} =
s_0s_1\cdots s_k$, it is not hard to compute the probability of
observing $\overline{\omega}$ along $\overline{s}$ on an HMM $H = ((S,
p), \Obs, o)$ with an initial state distribution $\pi$. Precisely,
\vspace{-5pt}
\begin{align}
  &\Pr (\overline{\omega}, \overline{s} | H)
  = \Pr (\overline{\omega} | \overline{s}, H) \times \Pr(\overline{s}, H)
    \nonumber
  \\
  & = [o (s_0, \omega_0) \! \cdots\!
     o (s_k, \omega_k)] \!\times\!
      [\pi (s_0) p (s_0, s_1) \! \cdots\!
     p (s_{k-1}, s_k)]
    \nonumber
  \\
  & = \pi (s_0) o (s_0, \omega_0) \cdot
    p (s_0, s_1)   \!\cdots\!
    p (s_{k-1}, s_k) o (s_k, \omega_k).
    \label{eqn:prob-obs-state}
\end{align}

\vspace{-3pt}
Since state sequences are not observable, we are interested in the
probability $\Pr (\overline{\omega} | H)$ for a given observation
sequence $\overline{\omega}$. Using~(\ref{eqn:prob-obs-state}), we have
 $\Pr (\overline{\omega} | H) = \sum_{\overline{s} \in S^{k+1}}
  \Pr (\overline{\omega}, \overline{s} | H)$.
But the summation has $|S|^{k+1}$ terms and is hence inefficient to
compute. An efficient algorithm is available to
compute the probability $\alpha_t (s)$ for the observation
sequence $\omega_0\omega_1\cdots\omega_t$ with the state $s$ at time
$t$~\cite{R:89:ATHMM}. Consider the following definition:
\begin{align}
  \alpha_0 (s) &= \pi(s) o (s, \omega_0)
  \label{hmm:basis}
  \\
  \alpha_{t+1} (s') &= \left[
                      \sum_{s \in S} \alpha_t (s) p (s, s')
                      \right]
                      o (s', \omega_{t+1}).
  \label{hmm:inductive}
\end{align}
Informally, $\alpha_0 (s)$ is the probability that the initial state is
$s$ with the observation $\omega_0$. By induction, $\alpha_t (s)$ is
the probability that the $t$-th state is $s$ with the observation sequence
$\omega_0 \omega_1 \cdots \omega_t$. The probability of observing
$\overline{\omega} = \omega_0 \omega_1 \cdots \omega_k$ is therefore
the sum of probabilities of observing $\overline{\omega}$ over all
states $s$.
Thus $\Pr (\overline{\omega} | H) = \sum_{s \in S} \alpha_k (s)$.

\section{Pufferfish Privacy Framework}
\label{section:pufferfish}

Differential privacy is a privacy framework for design and analysis of
data publishing mechanisms~\cite{DR:14:AFDP}. Let $\calX$
denote the set of \emph{data entries}. A \emph{data set} of size $n$
is an element in $\calX^n$.
Two data sets $\od, \od' \in \calX^n$
are \emph{neighbors} (written $\Delta (\od, \od') \leq 1$) if
$\od$ and $\od'$ are identical except for at most one data entry. A \emph{data
publishing mechanism} (or simply \emph{mechanism}) $\calM$ is a
randomized algorithm which takes a data set $\od$ as inputs. A
mechanism satisfies $\epsilon$-differential privacy if its output
distributions differ by at most the multiplicative factor $e^\epsilon$
on every neighboring data sets.

\begin{definition}
  Let $\epsilon \geq 0$. A mechanism $\calM$ is
  \emph{$\epsilon$-differentially private} if for all $r \in
  \textmd{range}(\calM)$ and data sets $\od, \od' \in \calX^n$ with
  $\Delta (\od, \od') \leq 1$, we have
$    \Pr (\calM (\od) = r) \leq e^{\epsilon} \Pr (\calM (\od') = r).$
\end{definition}

Intuitively, $\epsilon$-differential privacy ensures similar
output distributions on similar data sets. Limited differential
information about each
data entry is revealed and individual privacy is hence preserved.
Though, differential privacy makes no assumption nor uses any prior knowledge about data sets.
For data sets with correlated data entries, differential
privacy may reveal too much information about individuals. Consider,
for instance, a data set of family members. If a family member has
contracted a highly contagious disease, all family are likely to have
the same disease. In order to decide whether a specific family member
has contracted the disease, it suffices to determine whether
\emph{any} member has the disease. It appears that specific
information about an individual can be inferred from differential
information when data entries are correlated. Differential privacy may be
ineffective to preserve privacy in such circumstances~\cite{KM:11:NFLDP}.

Pufferfish is a Bayesian privacy framework which refines differential
privacy. Theorem~6.1 in \cite{KM:14:PFMPD}
shows how to define differential privacy equivalently in Pufferfish framework.
In Pufferfish privacy, a random variable
$\oD$ represents a data set drawn from a distribution $\theta \in
\bbfD$. The set $\bbfD$ of distributions formalizes prior knowledge
about data sets, such as whether data entries are independent or
correlated. Moreover, a set $\bbfS$ of \emph{secrets} and
a set $\bbfS_{\textmd{pairs}} \subseteq \bbfS \times \bbfS$ of
\emph{discriminative secret pairs} formalize the information to be
protected. 
A mechanism $\calM$ satisfies $\epsilon$-Pufferfish privacy if its output distributions
differ by at most the multiplicative factor $e^{\epsilon}$ when
conditioned on all the secret pairs.

\begin{definition}
  Let $\bbfS$ be a set of secrets, $\bbfS_{\textmd{pairs}} \subset
  \bbfS \times \bbfS$ a set of
  discriminative secret pairs, $\bbfD$ a set of data set distributions
  scenarios, and $\epsilon \geq 0$, a mechanism $\calM$ is
  \emph{$\epsilon$-Pufferfish private}
  if for all $r \in \textmd{range}(\calM)$, $(s_i, s_j) \in
    \bbfS_{\textmd{pairs}}$, $\theta \in \bbfD$ with $\Pr (s_i |
    \theta) \neq 0$ and $\Pr (s_j | \theta) \neq 0$, we have
    \vspace{-2pt}
    \[
      \Pr (\calM (\oD) = r | s_i, \theta) \leq
      e^\epsilon \Pr (\calM (\oD) = r | s_j, \theta)
    \]
    \vspace{-2pt}
    where $\oD$ is a random variable with the distribution $\theta$.
\end{definition}

In the definition, $\Pr (s_i | \theta) \neq 0$ and $\Pr (s_j | \theta)
\neq 0$ ensure the probabilities $\Pr (\calM (\oD)$ $=$ $r | s_i, \theta)$
and $\Pr (\calM (\oD) = r | s_j, \theta)$ are defined. Hence
$\Pr (\calM (\oD) = r | s, \theta)$ is the probability of observing
$r$ conditioned on the secret $s$ and the data set distribution $\theta$.
Informally, $\epsilon$-Pufferfish privacy ensures similar
output distributions on discriminative secrets and prior knowledge.
Since limited information is revealed from prior knowledge, each pair
of discriminative secrets is protected.

\section{Geometric Mechanism as Hidden Markov Model}
\label{section:hmm}

We first recall in Section~\ref{sub:GM} the definition of geometric mechanism, a well-known discrete mechanism for differential privacy. In Section~\ref{sub:dp-mc}, we then recall an example exploiting Markov chains to model geometric mechanisms, followed by our modeling formalism and Pufferfish privacy analysis using HMMs in Section~\ref{sub:PP}.
\vspace{-5pt}
\subsection{Geometric Mechanism}
\label{sub:GM}
Consider a simple data set with only two data entries. Each
entry denotes
whether an individual has a certain disease. Given such a
data set, we wish to know how many individuals contract the disease in
the data set. More generally, a \emph{counting} query returns the
number of entries satisfying a given predicate in a data set
$\od \in \calX^n$. The number of individuals contracting the disease
in a data set is hence a counting query.
Note that the difference of counting query results on neighboring data
sets is at most $1$.

Counting queries may reveal sensitive
information about individuals. For instance, suppose we know John's
record is in the data set. We immediately infer that John has
contracted the disease if the query answer is $2$.
In order to protect privacy, several mechanisms are designed to
answer counting queries.

\hide{
In fact, probabilistic inference may reveal
too much information in this case. If the query answer is $1$, we also
know that John has $50\%$ of chance to have the disease. Since the
population has much lower rate of contracting the disease, John's
privacy is undoubtedly intruded by our probabilistic inference.

In our simple scenario, the number of persons contracted
the disease in a data set is a counting query.
}

Consider a counting query $f: \calX^n \rightarrow \{ 0, 1,
\ldots, n \}$. Let $\alpha \in (0, 1)$.
The \emph{$\alpha$-geometric mechanism} $\calG_f$ for the counting
query $f$ on the data set $\od$
outputs $f(\od) + Y$ on a data set $\od$ where $Y$ is a random variable with the geometric
distribution~\cite{GRS:09:UUPM,GRS:12:UUPM}:
$
\Pr[Y = y] = \frac{1 - \alpha}{1 + \alpha}\alpha^{|y|}
\textmd{ for } y \in \bbfZ
$.
For any neighboring data sets $\od, \od' \in \calX^n$, recall that $|
f (\od) - f (\od')| \leq 1$. If $f (\od) = f (\od')$, the
$\alpha$-geometric mechanism has the same output distribution for $f$
on $\od$ and $\od'$.
If $|f (\od) - f (\od')|= 1$, it is easy to
conclude that $\Pr (\calG_f (\od) = r) \leq
e^{-\ln \alpha} \Pr (\calG_f (\od') = r)$ for any neighboring $\od,
\od'$ and $r \in \bbfZ$.
The $\alpha$-geometric mechanism is $-\ln \alpha$-differentially
private for any counting query $f$. To achieve $\epsilon$-differential
privacy,
one simply chooses $\alpha = e^{-{\epsilon}}$.

The range of the geometric mechanism is $\bbfZ$. It may give nonsensical
outputs such as negative integers for non-negative queries.
The \emph{truncated $\alpha$-geometric mechanism over $\{ 0, 1,
  \ldots, n \}$}
outputs $f (\od) + Z$ where $Z$ is a random variable with the
distribution:
\vspace{-5pt}
\[
\Pr[Z = z] =
\left\{
  \begin{array}{ll}
    0 & \textmd{ if } z < - f (x) \\\vspace{0.5ex}
    \frac{\alpha^{f (x)}}{1 + \alpha} & \textmd{ if } z = -f (x)\\ \vspace{0.5ex}
    \frac{1 - \alpha}{1 + \alpha}\alpha^{|z|} &
    \textmd{ if } -f (x) < z < n - f (x)\\\vspace{0.5ex}
    \frac{\alpha^{n - f (x)}}{1 + \alpha} & \textmd{ if } z = n-f (x)\\
    0 & \textmd{ if } z > n - f (x)
  \end{array}
\right.
\]

\vspace{-5pt}
Note the range of the truncated $\alpha$-geometric mechanism is
$\{ 0, 1, \ldots, n \}$. The truncated $\alpha$-geometric mechanism is
also $- \ln \alpha$-differentially
private for any counting query $f$.
We will study several examples of this mechanism to get a better understanding of Pufferfish privacy
and how we use models to analyze it.

\subsection{Differential Privacy Using Markov Chains}\label{sub:dp-mc}
We present a simple example taking from~\cite{LWZ:18:MCDPP}, slightly adapted for analyzing
different models, i.e., the Markov chain and the hidden Markov model.
\vspace{-25pt}
\begin{figure}[htbp]
  \begin{subfigure}[c]{0.33\columnwidth}
  \renewcommand\arraystretch{1.2}
  \linespread{2}
  \parbox[][3cm][c]{\linewidth}{
  \centering
    \[
    \begin{array}{c|c|c|c|c|}
      \multicolumn{2}{c}{} &
      \multicolumn{3}{c}{output} \\
      \cline{2-5}
      &\hphantom{1/3} & \tilde{0} & \tilde{1} & \tilde{2} \\
      \cline{2-5}
      \multirow{3}*{\rotatebox{90}{$\mathit{input}$}}
      & 0 & {2}/{3}  & {1}/{6} & {1}/{6} \\
      \cline{2-5}
      & 1 & {1}/{3}  & {1}/{3} & {1}/{3} \\
      \cline{2-5}
      & 2 & {1}/{6}  & {1}/{6} & {2}/{3} \\
      \cline{2-5}
    \end{array}
    \]}
    \vspace{-5pt}
    \caption{$\frac{1}{2}$-Geometric Mechanism}
    \label{figure:geometric-mechanism-table}
  \end{subfigure}
  \begin{subfigure}[c]{.33\columnwidth}
  \parbox[][3cm][c]{\linewidth}{
  \centering
    \begin{tikzpicture}[->,>=stealth',shorten >=1pt,auto,node
      distance=2cm,node/.style={circle,draw}]
      \node[node ]
            (i0) at ( -2,  1) { $0$ };
      \node[node ]
            (i1) at ( -2,  0) { $1$ };
      \node[node ]
            (i2) at ( -2, -1) { $2$ };

      \node[node ]
            (o0) at (  0,  1) { ${\tilde{0}}$ };
      \node[node ]
            (o1) at (  0,  0) { ${\tilde{1}}$ };
      \node[node ]
            (o2) at (  0, -1) { ${\tilde{2}}$ };
      \draw[->,very thick] (i0) -- (o0);   
      \draw[->,ultra thin]  (i0) -- (o1);   
      \draw[->,ultra thin]  (i0) -- (o2);   

      \draw[->,semithick] (i1) -- (o0);               
      \draw[->,semithick] (i1) -- (o1);               
      \draw[->,semithick] (i1) -- (o2);               

      \draw[->,ultra thin]  (i2) -- (o0);   
      \draw[->,ultra thin]  (i2) -- (o1);   
      \draw[->,very thick] (i2) -- (o2);   

      \node at (0.9, 1.3) { $\frac{1}{6}$ };
      \draw[->,ultra thin] (0.5,  1) -- (1.3,  1);
      \node at (0.9, 0.3) { $\frac{1}{3}$ };
      \draw[->,semithick]  (0.5,  0) -- (1.3,  0);
      \node at (0.9,-0.7) { $\frac{2}{3}$ };
      \draw[->,very thick] (0.5, -1) -- (1.3, -1);
      \end{tikzpicture}}
    \caption{Markov Chain}
    \label{figure:geometric-mechanism-mc}
  \end{subfigure}
   \begin{subfigure}[c]{.3\columnwidth}
   \parbox[][3cm][c]{\linewidth}{
  \centering
    \begin{tikzpicture}[->,>=stealth',shorten >=1pt,auto,node
      distance=2cm,node/.style={circle,draw}]
      \node[node , label={ 10:2/3:$\tilde{0}$},label={0:1/6:$\tilde{1}$},label={-10:1/6:$\tilde{2}$}]
            (i0) at ( -2,  1) { $0$ };
      \node[node , label={ 10:1/3:$\tilde{0}$},label={0:1/3:$\tilde{1}$},label={-10:1/3:$\tilde{2}$}]
            (i1) at ( -2,  0) { $1$ };
      \node[node , label={ 10:1/6:$\tilde{0}$},label={0:1/6:$\tilde{1}$},label={-10:2/3:$\tilde{2}$}]
            (i2) at ( -2, -1) { $2$ };
      \end{tikzpicture}}
    \caption{Hidden Markov Model}
    \label{figure:geometric-mechanism-hmm2}
  \end{subfigure}
  \vspace{-2pt}
  \caption{Truncated $\frac{1}{2}$-Geometric Mechanism}
  \label{figure:geometric-mechanism}
\end{figure}
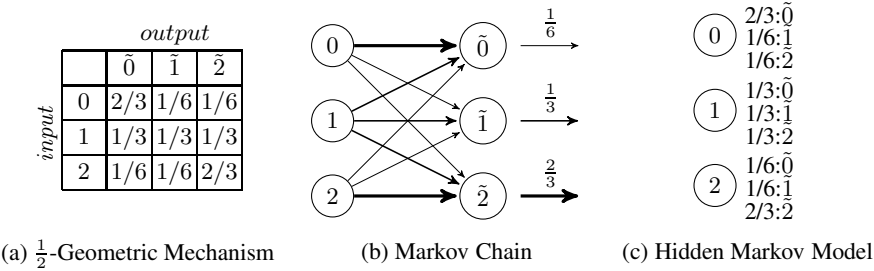
\vspace{-25pt}
\begin{example}
To see how differential privacy works, consider the truncated
$\frac{1}{2}$-geometric mechanism
(Fig.~\ref{figure:geometric-mechanism-table}). In the table, we
consider a counting query $f : \calX^2 \rightarrow \{ 0, 1, 2 \}$.
For any data set $\od$, the mechanism outputs $j$ when $f
(\od) = i$ with probability indicated at the $(i, j)$-entry in the
table. For instance, the mechanism outputs $\tilde{0}$, $\tilde{1}$,
and $\tilde{2}$ with
probabilities $\frac{2}{3}$, $\frac{1}{6}$, and $\frac{1}{6}$
respectively when $f (\od) = 0$.

Let $f$ be the query counting the number of individuals contracting
a disease. Consider a data set $\od$ whose two members (including John)
have contracted the disease. The number of individuals contracting the
disease is $2$ and hence $f (\od) = 2$.
From the table in Fig.~\ref{figure:geometric-mechanism-table},
we see the mechanism answers
$\tilde{0}$, $\tilde{1}$, and $\tilde{2}$ with probabilities
$\frac{1}{6}$, $\frac{1}{6}$, and $\frac{2}{3}$ respectively.
Suppose we obtain another data set $\od'$ by replacing John with an
individual who does not contract the disease. The number of
individuals contracting the disease for the new data set is $1$ and
thus $f (\od') = 1$. Then, the mechanism answers
$\tilde{0}$, $\tilde{1}$, and $\tilde{2}$ with the probability
$\frac{1}{3}$.

The probabilities of observing $\tilde{0}$ on the data sets
$\od$ and $\od'$ are respectively $\frac{1}{6}$ and
$\frac{1}{3}$. They differ by the multiplicative factor $2$. For
other outputs, their observation probabilities are also bounded by the
same factor. The truncated $\frac{1}{2}$-geometric mechanism is hence
$\ln (2)$-differentially private.

\hide{
If an attacker queries the number of individuals contracting the
disease through the truncated $\frac{1}{2}$-geometric mechanism,
he will get the answer ${2}$ with probability at least
$\frac{1}{3}$, and at most $\frac{2}{3}$ regardless of John's record.
For any two similar data sets, the mechanism always have similar
output distributions. Information revealed to attackers is therefore
limited.
}

In order to formally analyze privacy mechanisms, we specify them
as probabilistic models. Fig.~\ref{figure:geometric-mechanism-mc} shows
a Markov chain for the truncated $\frac{1}{2}$-geometric mechanism.
We straightly turn inputs and outputs of the table in Fig.~\ref{figure:geometric-mechanism-table} into states of the Markov chain
and output probabilities into transition probabilities.
In the figure, thin arrows denote transitions with probability
$\frac{1}{6}$; medium arrows denote transitions with probability
$\frac{1}{3}$; thick arrows denote transitions with probability
$\frac{2}{3}$. For instance, state $0$ can transit to state
$\tilde{0}$ with probability $\frac{2}{3}$ while it can transit to
the state $\tilde{1}$ with probability $\frac{1}{6}$.
\hfill$\blacksquare$
\end{example}
The Markov chain model is straightforward but can become hazy for complicated privacy mechanism. We next discuss how to use an HMM to model the mechanism.

\subsection{Pufferfish Privacy Using Hidden Markov Models}
\label{sub:PP}
We denote data sets as states
and possible outputs of the mechanism are denoted by
observations. The transition
distribution stimulates the randomized privacy mechanism performed
on data sets. 
Distributions of data
sets are denoted by initial information states. Privacy
analysis can then be performed by comparing observation probabilities
from the two initial information states. We illustrate the ideas
in examples.
\begin{example}
\label{ex:hmm1}
Fig.~\ref{figure:geometric-mechanism-hmm2} gives an HMM for
the truncated $\frac{1}{2}$-geometric mechanism. For any counting
query $f$ from
$\calX^2$ to $\{ 0, 1, 2 \}$, it suffices to represent each $\od \in
\calX^2$ by $f (\od)$ because the mechanism only depends on $f (\od)$.
The order of entries, for instance, is irrelevant to the mechanism.
We hence have the states $0$, $1$ and $2$ denoting the set $\{ f (\od)
: \od \in \calX^2 \}$ in the figure.
Let $\{\tilde{0}, \tilde{1}, \tilde{2} \}$ be the set
of observations. We encode output probabilities into observation probabilities at states.
At state $0$, for instance, $\tilde{0}$, $\tilde{1}$, $\tilde{2}$ can all be observed
with probability $\frac{2}{3}$, $\frac{1}{6}$, $\frac{1}{6}$ respectively.
It is obvious that the number of states are reduced by half  compared with the Markov chain. Generally, HMMs allow multiple observations to show at one single state,
which leads to smaller models.

Fix an order for states,
say, $0, 1, 2$. An information state
can be represented by an element in $[0, 1]^3$.
In differential privacy,
we would like to analyze probabilities of every observation from
neighboring data sets. For counting queries,
neighboring data sets can change query results by at most $1$. Let
$\od$ be a data set. Consider the initial information state
$\pi = (0, 0, 1)$ corresponding to $f (\od) = 2$. For
any neighbor $\od'$ of $\od$, we have $f (\od') = 2$ or $f (\od') =
1$. It  suffices to consider corresponding information states
$\pi$ or $\tau = (0, 1, 0)$. Let's compare
the probability of  observing $\omega = \tilde{1}$ from
information states $\pi$ and $\tau$. Starting from $\pi$, we have
$\alpha_0 = \pi$ and probabilities of $\frac{1}{6}$, $\frac{1}{3}$ and $\frac{1}{6}$
respectively observing $\tilde{1}$ at each state. So the probability of observing
$\omega$ is $\frac{1}{6}$.
On the other hand, we have $\alpha_0 = \tau$ and the probability of observing
$\omega$ is $\frac{1}{3}$.
Similarly, one can easily check the probabilities of
observing $\tilde{0}$ and $\tilde{2}$ on any neighboring
data sets and the ratio of one probability over the other one under the same observation
will not be more than $2$.
$\hfill\blacksquare$
\end{example}

\hide{
Figure~\ref{figure:geometric-mechanism-hmm2} shows a smaller hidden Markov model for
the mechanism. The intuition of this model is quite easy to see after we analyze the above models.
There are only a half of states without transitions in this model and from
each state, one of the observations $\{\tilde{0}, \tilde{1}, \tilde{2}\}$can be obtained with different
observing probabilities. This is a simple example indicating that HMMs may result in models of small size.
}


Differential privacy provides a framework for quantitative privacy analysis.
The framework ensures similar output distributions regardless of the information about an arbitrary individual.
In other words, if an attacker gets certain prior knowledge about the data sets,
chances are that differential privacy will underestimate privacy risks. Since all data entries are correlated, replacing one data
entry does not yield feasible data sets with correlated entries.
Consequently, it is questionable to compare
output distributions on data sets differing in only one entry. Instead, this is the scenario where Pufferfish privacy should be applied.

\vspace{-5pt}
\begin{example}
\label{ex:hmm2}
Consider a data set about contracting a highly contagious disease 
containing John and a
family member he lives with. An attacker wishes to know if John has contracted the
disease. Since the data set keeps information on the contagious disease about two
family members, an attacker immediately deduces that the number of
individuals contracting the disease can only be $0$ or $2$. The
attacker hence can infer whether John has the disease by counting the
number of individuals contracting the disease.

Suppose a data curator tries to protect John's privacy by employing
the truncated $\frac{1}{2}$-geometric mechanism
(Fig.~\ref{figure:geometric-mechanism}).
We analyze this mechanism formally in the Pufferfish framework. Let the set of data entries $\calX = \{ 0, 1
\}$ and there are four possible data sets in $\calX^2$. For any $0 < p < 1$,
define the data set distribution $\theta_p : \calX^2 \rightarrow [0,
1]$ as follows. $\theta_p  (0, 0) = 1 - p$, $\theta_p (1, 1) = p$, and
$\theta_p (0, 1) = \theta_p (1, 0) = 0$.
Consider the distribution set $\bbfD = \{ \theta_p : 0 < p < 1 \}$. Note
that infeasible data sets are not in the support of $\theta_p$.

Assume John's entry is in the data set. Define the set
of secrets $\bbfS = \{ c, nc \}$ where $c$ denotes that John has
contracted the disease and $nc$ denotes otherwise. Our set of
discriminative secret pairs $\bbfS_{\textmd{pairs}}$ is $\{ (c, nc),
(nc, c) \}$. That is, we would like to compare probabilities of
all outcomes when John has the disease or not.

When John has not contracted the disease, the only possible data set
is $(0, 0)$ by the distribution $\theta_p$.
The probability of observing $\tilde{0}$ therefore is $\frac{2}{3}$
(Fig.~\ref{figure:geometric-mechanism-table}). When John has the
disease, the data set $(0, 0)$ is not possible under the condition of
the secret and the distribution $\theta_p$.
The only possible data set is $(1, 1)$.
The probability of observing $\tilde{0}$ is $\frac{1}{6}$.
Now we have
$
  \frac{2}{3} = \Pr (\calG_f (\oD) = \tilde{0} | nc, \theta_p)
  \not\leq
  2 \times \frac{1}{6} =
  2 \times \Pr (\calG_f (\oD) = \tilde{0} | c, \theta_p).
$
We conclude the truncated $\frac{1}{2}$-geometric mechanism does not
conform to $\ln (2)$-Pufferfish privacy. Instead, it satisfies $\ln (4)$-Pufferfish privacy.
$\hfill\blacksquare$
\end{example}


\begin{wraptable}[7]{r}{8cm}
\vspace{-20pt}
  \caption{Pufferfish Analysis of $\frac{1}{2}$-Geometric Mechanism}
  \label{table:Pufferfish-geometric-mechanism}
  \centering
  \renewcommand\arraystretch{1.4}
\vspace{-10pt}
  \[
    \begin{array}{|c|c|c|c|}
    \hline
      \textmd{Data Sets$\backslash$Observations}& \tilde{0} & \tilde{1} & \tilde{2} \\
      \hline
      \textmd{without John's record}
      & \frac{p^2 - 4p + 4}{6}
      & \frac{-2p^2 + 2p + 1}{6}
      & \frac{p^2 + 2p + 1}{6}
      \\
      \hline
      \textmd{with John's record}
      & \frac{4-3p}{12-6p}
      & \frac{4-3p}{12-6p}
      & \frac{2}{6-3p}
      \\
      \hline
    \end{array}
  \]
\end{wraptable}
\vspace{-5pt}

With the formal
model (Fig.~\ref{figure:geometric-mechanism-hmm2}), it is easy to
perform privacy analysis in the Pufferfish framework. More precisely,
the underlying Markov chain along with observation distribution specify the privacy mechanism on
input data sets. Prior knowledge about data sets is nothing but
distributions of them. Since data sets are represented by various
states, prior knowledge is naturally formalized as initial information
states in HMMs. For Pufferfish privacy analysis, we
again compare observation probabilities from initial information
states conditioned on secret pairs. The standard algorithm for HMMs allows us to perform more refined privacy analysis.
Besides, it is interesting to observe the striking similarity between
the Pufferfish privacy framework and HMMs. In both
cases, input data sets are unknown but specified by
distributions. Information can only be released by observations
because inputs and hence computation are hidden from external
attackers or observers. Pufferfish privacy analysis with prior
knowledge is hence closely related to observation probability analysis
from information states. Such similarities can easily be identified in
the examples.

\hide{
\noindent
\emph{Example 4.5}
For a data set $\oD \in \calX^2$ drawn by the distribution
$\theta_p$, we have $f (\oD) = 0$ with probability $1 - p$ and $f
(\oD) = 2$ with probability $p$. The data set distribution $\theta_p$
thus corresponds to the information state $(1 - p, 0, p, 0, 0, 0)$.
Given the information state, let us compute the probability of
observing $\blank\tilde{0}$ in two scenarios. If John has not
contracted the disease, we start from the initial information state
$\pi = (\frac{1 - p}{1 - p}, 0, 0, 0, 0, 0) = (1, 0, 0, 0, 0, 0)$ by
$\theta_p$. Using (\ref{hmm:inductive}), we have $\alpha_0 = \pi$ and
$\alpha_1 = (0, 0, 0, \frac{2}{3}, 0, 0)$. Similarly, we start from
the initial information state $\tau = (0, 0, \frac{p}{p}, 0, 0, 0) =
(0, 0, 1, 0, 0, 0)$ when John has the disease. We have $\alpha_0 =
\tau$ and $\alpha_1 = (0, 0, 0, \frac{1}{6}, 0, 0)$ in
(\ref{hmm:inductive}). The probabilities of observing
$\blank\tilde{0}$ from
$\pi$ and $\tau$ are hence $\frac{2}{3}$ and $\frac{1}{6}$
respectively. Using the formal model, we reach the same conclusion as
in privacy analysis.
$\hfill\blacksquare$

When an attacker has prior knowledge about data sets, we have seen why
differential privacy may underestimate information leak. Prior
knowledge however does not necessarily  lead to privacy leak.
In the following example, we demonstrate that certain prior knowledge
is indeed not harmful to privacy under the Pufferfish framework using
our formal model.}

\begin{example}
\label{ex:hmm4}
Consider a non-contagious disease. An attacker may know that
contracting the disease is an independent event with probability
$p$. Even though the attacker does not know how many individuals have
the disease exactly, he infers that the number of individuals
contracting the disease is $0$, $1$, and $2$ with probabilities
$(1-p)^2$, $2p(1-p)$, and $p^2$ respectively.
The prior knowledge corresponds to the initial information state
$\pi = ((1-p)^2, 2p(1-p), p^2)$ in
Fig.~\ref{figure:geometric-mechanism-hmm2}.
Assume John has contracted the disease. We would
like to compare probabilities of observations $\tilde{0}$,
$\tilde{1}$, and $\tilde{2}$ given the prior knowledge and
the presence or absence of John's record.

Suppose John's record is indeed in the data set. Since John has
the disease, the number of individuals contracting the disease cannot
be $0$. By the prior knowledge, one can easily obtain the initial information state
$\pi = (0$, $\frac{2p(1-p)}{2p(1-p) + p^2}$,
$\frac{p^2}{2p(1-p) + p^2}$) = ($0$, $\frac{2-2p}{2-p}$,
$\frac{p}{2-p})$. If John's record is not in the data set, the initial information state remains as
$\tau = ((1-p)^2, 2p(1-p), p^2)$. Then one can compute all the observation probabilities starting from $\pi$ and $\tau$ respectively, which are summarized in 
Table~\ref{table:Pufferfish-geometric-mechanism}:

For the observation $\tilde{0}$, it is not hard to check $\frac{1}{2}
\times \frac{4-3p}{12-6p} \leq \frac{p^2 - 4p + 4}{6} \leq 2 \times
\frac{4-3p}{12-6p}$ for any $0 < p < 1$. Similarly, we have
$\frac{1}{2} \times
\frac{4-3p}{12-6p} \leq \frac{-2p^2 + 2p + 1}{6} \leq 2 \times
\frac{4-3p}{12-6p}$ and $\frac{1}{2} \times \frac{2}{6-3p} \leq
\frac{p^2 + 2p + 1}{6} \leq 2 \times \frac{2}{6-3p}$  for observations
$\tilde{1}$ and $\tilde{2}$ respectively. Therefore, the
truncated $\frac{1}{2}$-geometric mechanism satisfies $\ln(2)$-Pufferfish
privacy when contracting the disease is \emph{independent}.
$\hfill\blacksquare$
\end{example}
\vspace{-5pt}
The above example demonstrates that certain prior knowledge, such as independence
of data entries,
is indeed not harmful to privacy under the Pufferfish framework.
In~\cite{KM:14:PFMPD}, it is shown that differential privacy is
subsumed by Pufferfish privacy (Theorem~6.1) under independence assumptions. The above example is also an
instance of the general theorem but formalized in an HMM.

\section{Pufferfish Privacy Verification}
\label{section:verification}
In this section, we formally define the verification problem for Pufferfish privacy and give
the computation complexity results in Section~\ref{subsection:complexity}.
Then we propose an algorithm to solve the problem in Section~\ref{subsection:checking-pufferfish}.

\subsection{Complexity of Pufferfish Privacy Problem}
\label{subsection:complexity}

We model the general Pufferfish privacy problems into HMMs and
the goal is to check whether the privacy is preserved. First, we define the \emph{Pufferfish verification problem}:

\begin{definition}
Given a set of secrets $\bbfS$,
a set of discriminative secret pairs $\bbfS_{\textmd{pairs}}$, a set of data evolution
scenarios $\bbfD$ , $\epsilon > 0$, along with mechanism  $\calM$ in a hidden Markov model $H = (K, \Obs, o)$,
where probability distributions are all discrete.
Deciding whether $\calM$  satisfies $\epsilon$-Pufferfish privacy under $(\bbfS$, $\bbfS_{\textmd{pairs}}$,
$\bbfD)$ is the \emph{Pufferfish verification problem}.
\end{definition}

The modeling intuition for $H$ is to use states and transitions to model the data sets and operations in the mechanism $\calM$,
obtain initial distribution pairs according to prior knowledge $\bbfD$ and discriminative secrets $\bbfS_{\textmd{pairs}}$,
and set outputs as observations in states. Then the goal turns into checking whether the probabilities under the same observation sequence
are mathematically similar, i.e., differ by at most the multiplicative factor $e^{\epsilon}$, for every distribution pair and every observation sequence.
Therefore, our task is to find the observation sequence and distribution pair that make the observing probabilities differ the most.
That is, in order to satisfy Pufferfish privacy, for every observation sequence $\overline{\omega}=\omega_1\omega_2\ldots$, secret pair $(s_i, s_j) \in
\bbfS_{\textmd{pairs}}$ and $\theta \in \bbfD$, one should have
    \begin{equation}\label{max1}
     \max_{\overline{\omega},(s_i, s_j) ,\theta}
    { \Pr (\calM (\oD) = \overline{\omega}| s_i, \theta) - e^\epsilon \Pr (\calM (\oD) = \overline{\omega}| s_j, \theta) }
  \end{equation}
  \begin{equation}\label{max2}
     \max_{\overline{\omega},(s_i, s_j),\theta}
     { \Pr (\calM (\oD) = \overline{\omega}| s_j, \theta) - e^\epsilon \Pr (\calM (\oD) = \overline{\omega}| s_i, \theta) }
  \end{equation}
no more than $0$. However, by
showing a reduction from the classic Boolean Satisfiability Problem~\cite{PCT:87:CMDP},
this problem is proved to be NP-hard (in Appendix 1 ):

\begin{theorem}\label{theorem1}
  The Pufferfish verification problem is NP-hard.
\end{theorem}

To the best of our knowledge, this is the first complexity result for the
Pufferfish verification problem. Note that differential privacy is subsumed
by Pufferfish privacy. Barthe et al.~\cite{BCJSV:20:ddppfio} show undecidability
results for differential privacy mechanisms with continuous noise. Instead, we focus on
Pufferfish privacy with discrete state space in HMMs.
The complexity bound is lower if
more simple models such as Markov chains are used.
However some discrete mechanisms in differential privacy, such as Above Threshold,
can hardly be modeled in Markov chains~\cite{LWZ:18:MCDPP}.  

\subsection{Verifying Pufferfish Privacy}
\label{subsection:checking-pufferfish}
Given the complexity lower bound in the previous section,
next goal is to develop an algorithm to verify
$\epsilon$-Pufferfish privacy on any given HMM.
We employ Satisfiability Modulo Theories (SMT) solvers
in our algorithm. For all observation sequences of length $k$, we will
construct an SMT query to find a sequence violating
$\epsilon$-Pufferfish privacy. If no such sequence can be found, the
given HMM satisfies $\epsilon$-Pufferfish privacy for all observation
sequences of length $k$.

\begin{algorithm}
  \begin{algorithmic}[1]
    \Require{$H = ((S, p), \Obs, o)$: a hidden Markov model;
      $\pi, \tau$: state distributions on $S$; $c$: a
      non-negative real number; $k$: a positive integer}
    \Ensure{An SMT query $q$ such that $q$ is unsatisfiable iff
      $\Pr (\overline{\omega} | \pi, H) \leq c \cdot
       \Pr (\overline{\omega} | \tau, H)$ for every observation
       sequences $\overline{\omega}$ of length $k$}
     \Function{PufferfishCheck}{$H$, $\pi_0$, $\pi_1$, $c$, $k$}
       \For{$s \in S$}
         \State{$\alpha_0(s) \leftarrow
           \textmd{\textsc{Product}} (\pi(s),
           \textmd{\textsc{Select}}(\mathsf{w_0}, \Omega, o(s, \bullet)))$}
         \State{$\beta_0(s) \leftarrow
           \textmd{\textsc{Product}} (\tau(s),
           \textmd{\textsc{Select}}(\mathsf{w_0}, \Omega, o(s, \bullet)))$}
       \EndFor
       \For{$t \leftarrow 1$ \textbf{to} $k - 1$}
         \For{$s' \in S$}
           \State{$\alpha_t(s') \leftarrow
             \textmd{\textsc{Product}} (
             \textmd{\textsc{Dot}} (\alpha_{t-1}, p (\bullet, s')),
             \newline\indent\indent\textmd{\textsc{Select}} (\mathsf{w_t}, \Omega, o(s', \bullet)))
             $}
           \State{$\beta_t(s') \leftarrow
             \textmd{\textsc{Product}} (
             \textmd{\textsc{Dot}} (\beta_{t-1}, p (\bullet, s')),
             \newline\indent\indent\textmd{\textsc{Select}} (\mathsf{w_t}, \Omega, o(s', \bullet)))
             $}
         \EndFor
       \EndFor
       \State{\Return{
           $\textmd{\textsc{Gt}} (
           \textmd{\textsc{Sum}} (\alpha_{k-1}),
           \textmd{\textsc{Product}} (c,
           \textmd{\textsc{Sum}} (\beta_{k-1})))\wedge
           \bigwedge_{t=0}^{k-1} \mathsf{w_t} \in \Obs
           $
       }}
     \EndFunction
 \end{algorithmic}
 \caption{Pufferfish Check}
 \label{algorithm:pufferfish-check}
\end{algorithm}

Let $H = ((S, p), \Obs, o)$ be an HMM, $\pi, \tau$ two initial 
distributions on $S$, $c \geq 0$ a real number, and $k$ a positive
integer. With a fixed observation sequence $\overline{\obs}$,
computing the probability $\Pr (\overline{\obs} | \pi,
H)$ 
can be done in
polynomial time~\cite{R:89:ATHMM}. To check if
$\Pr (\overline{\obs} | \pi, H) > c \cdot \Pr (\overline{\obs} |
\tau, H)$ for any fixed observation sequence $\overline{\obs}$, one
simply computes the respective probabilities and then checks the inequality.

Our algorithm exploits the efficient algorithm of HMMs for computing the
probability of observation sequences. Rather than a fixed observation
sequence, we declare $k$ SMT variables $\mathsf{w_0}, \mathsf{w_1},
\ldots, \mathsf{w_{k-1}}$ for observations at each step. The
observation at each step is determined by one of the $k$ variables.
Let $\Omega = \{ \omega_1$, $\omega_2$, $\ldots$, $\omega_m \}$ be the set
of observations. We define the SMT expression
$\textmd{\textsc{Select}}$ ($\mathsf{w}$, $\{ \obs_1$, $\obs_2$, $\ldots$,
$\obs_m\}$, $o (s, \bullet))$  equal to $o (s, \obs)$ when the SMT
variable $\mathsf{w}$ is $\obs \in \Omega$. It is straightforward to
formulate by the SMT $\mathsf{ite}$ (if-then-else) expression:
\begin{align*}
  \mathsf{ite} (\mathsf{w} = \obs_1, o (s, \obs_1),
  &\mathsf{ite} (\mathsf{w} = \obs_2, o (s, \obs_2),
  \ldots,
  \mathsf{ite} (\mathsf{w} = \obs_m, o (s, \obs_m), \mathsf{w})
  \ldots))
\end{align*}

Using $\textmd{\textsc{Select}} (\mathsf{w}, \{ \obs_1, \obs_2,
\ldots, \obs_m\}, o (s, \bullet))$, we construct an SMT
expression to compute $\Pr (\overline{\mathsf{w}} | \pi, H)$ where
$\overline{\mathsf{w}}$ is a sequence of SMT variables ranging over
the observations $\Obs$ (Algorithm~\ref{algorithm:pufferfish-check}).
Recall the equations (\ref{hmm:basis}) and (\ref{hmm:inductive}). We
simply replace the expression $o (s, \obs)$ with the new one
$\textmd{\textsc{Select}} $$(\mathsf{w}$$, \{ \obs_1,$ $\obs_2,$ $\ldots,$
$\obs_m\},$ $o$ $(s,$ $\bullet))$ to leave the observation determined by the
SMT variable $\mathsf{w}$. In the algorithm, we also use auxiliary
functions.
$\textmd{\textsc{Product}}(\mathit{smtExp}_0, 
\ldots, \mathit{smtExp}_m)$
returns the SMT expression denoting the product of
$\mathit{smtExp}_0, \ldots,$
$\mathit{smtExp}_m$. Similarly,
$\textmd{\textsc{Sum}}(\mathit{smtExp}_0,\ldots,$ 
$ \mathit{smtExp}_m)$ returns the SMT expression for the sum
of $\mathit{smtExp}_0, \ldots,$
$\mathit{smtExp}_m$. $\textmd{\textsc{Gt}} (\mathit{smtExp}_0,
\mathit{smtE}$ $\mathit{xp}_1)$ returns the SMT expression for
$\mathit{smtExp}_0$ greater than $\mathit{smtExp}_1$.
Finally,
$\textmd{\textsc{Dot}}$ $([\mathsf{a_0}$, $\mathsf{a_1}$, $\ldots,$
$\mathsf{a_n}], [\mathsf{b_0}, \mathsf{b_1}, \ldots, \mathsf{b_n}])$
returns the SMT expression for the inner product of the two lists of
SMT expressions, namely,
$
\textmd{\textsc{Sum}} (
  \textmd{\textsc{Product}}(\mathsf{a_0}, \mathsf{b_0}),
  \ldots,
  \textmd{\textsc{Product}}(\mathsf{a_n}, \mathsf{b_n})
  ).
$

Algorithm~\ref{algorithm:pufferfish-check} is summarized in the
following theorem.

\begin{theorem}
  Let $H = ((S, p), \Obs, o)$ be a hidden Markov model, $\pi, \tau$ state
  distributions on $S$, $c > 0$ a real number, and $k > 0$ an
  integer. Algorithm~\ref{algorithm:pufferfish-check} returns an SMT
  query such that the query is unsatisfiable iff
  $\Pr (\overline{\omega} | \pi, H) \leq c \cdot
  \Pr (\overline{\omega} | \tau, H)$ for every observation
  sequence $\overline{\omega}$ of length $k$.
\end{theorem}
\vspace{-5pt}

In practice, the integer $k$ depends
on the length of observation sequence we want to make sure to satisfy Pufferfish privacy.
For instance, in the model of Fig.~\ref{figure:geometric-mechanism-hmm2}, the maximal length
of observation sequence is 1 and thus $k=1$. If there exist
cycles in models such as Fig.~\ref{figure:hmm-above-threshold}, which implies loops in the mechanisms,
$k$ should keep increasing (and stop before a set value) in order to examine outputs of different lengths.

\hide{The factors that determine the speed of our algorithm are mainly the size
of model $H$, the length $k$ of observation sequence and the computing
power of SMT solver. We have implemented our algorithm using \zpython. All examples in
Section~\ref{section:hmm} are checked with our implementation. In
\textit{Example~4.5}, contracting the disease is independent with
probability $p > 0$. Even though the query contains non-linear
constraints over real numbers, \zpython still proves that the truncated
$\frac{1}{2}$-geometric mechanism satisfies $\ln(2)$-Pufferfish
privacy for any probability $0 < p < 1$.}

\section{Pufferfish Privacy Verifier: FAIER}
\label{section:case-studies}

We implement our verification tool and present experimental results in Subsection~\ref{subsection:evaluation}.
For the well-known differential privacy
mechanisms  Noisy Max and Above Threshold,
we provide modeling details in HMMs and verify the privacy
wrt. several Pufferfish privacy scenarios in Subsection~\ref{subsection:noisymax} and \ref{subsection:above-threshold}, accordingly.
\vspace{-5pt}
\subsection{Evaluation for FAIER}
\label{subsection:evaluation}
We implement our verification algorithm (Algorithm~\ref{algorithm:pufferfish-check}) into the tool FAIER,
which is the pufferFish privAcy verifIER.
It is implemented in \textbf{C++} environment with the SMT solver $Z3$~\cite{MB:08:z3} and
we performed all experiments on an Intel(R) Core i7-8750H
@ 2.20GHz CPU machine with 4 GB memory and 4 cores in the virtual machine. 
All the examples in this paper have been verified.

The inputs for our tool include an HMM $H$ of the mechanism to be verified,
distribution pair ($\pi$,$\tau$) on states in $H$,
a non-negative real number $c$ indicating
the privacy budget and an input $k$ specifying the length of observation
sequences. Note that unknown parameters are also allowed in the SMT formulae, which can encode certain prior knowledge or data sets distributions.

\vspace{-20pt}
\begin{table*}[h]
\caption{Experiment results: \goodmark~indicates the property holds, and \badmark~not.  }
\label{table:experiment-results}
\begin{tabular}{m{3cm}<{\centering}|m{3cm}<{\centering}|m{2cm}<{\centering}|m{3.6cm}<{\centering}}
     \multirow{2}{*}{Mechanism} & \multirow{2}{*}{Privacy scenario} & \multicolumn{2}{c}{Result}  \\ \cline{3-4}
     & & Query answer & Counterexample  \\
      \hline
     \multirowcell{4}{Truncated \\$\frac{1}{2}$ -geometric \\Mechanism} & $\ln (2)$-differential privacy (Ex.~\ref{ex:hmm1}) & \goodmark & \\ \cline{2-4}
     & $\ln (2)$-pufferfish privacy (Ex.~\ref{ex:hmm2})& \badmark &$\tilde{2}$\\ \cline{2-4}
     & $\ln (2)$-pufferfish privacy (Ex.~\ref{ex:hmm4})& \goodmark &\\\cline{2-4}
     \hline
     \multirowcell{2}{Discrete \\Noisy Max \\(Algorithm~\ref{algorithm:noisy-max})} &$\ln (2)$-pufferfish privacy (Ex.~\ref{ex:nm1}) & \goodmark  &\\\cline{2-4}
     & $\ln (2)$-pufferfish privacy (Ex.~\ref{ex:nm2})& \badmark &$\bot,\tilde{3}$; $p_A=p_B=p_C= \frac{1}{2}$\\
     \hline
     Above Threshold Algorithm (Algorithm~\ref{algorithm:above-threshold})& $4\ln (2)$-differential privacy & \badmark &
     $\blank, 01, \bot, 12, \bot, 12, \bot$, $12$, $\bot$, $21, \top$\\
\end{tabular}
\end{table*}
\vspace{-10pt}

We summarize the experiment results in this paper for pufferfish privacy, as well as differential privacy in Table~\ref{table:experiment-results}. FAIER has the following outputs:
\begin{itemize}
\item \emph{Counterexample:} If the privacy condition does not hold (marked by \badmark),
FAIER will return a witnessing observation sequence leading to the violation.
\item \emph{Parameter Synthesis:}
If there exist unknown parameters in the model, such as the infection rate $p$ for some disease, a value will be synthesized for the counterexample. See Ex.~\ref{ex:nm2} where counterexample is found when $p_A,p_B,p_C$ are equal to $\frac{1}{2}$; Or, no value can be found if the privacy is always preserved. See Ex.~\ref{ex:nm1}.
\item \emph{\goodmark} is returned if the privacy is preserved.
\end{itemize}
Note that if there exists a loop in the model, the bound $k$ should continue to increase when an 'UNSAT' is returned. Specially, the bound is set at a maximum of $15$ for Above Threshold.
It may happen that FAIER does not terminate since some nonlinear constraints are too complicated for $Z3$, such as Ex.~\ref{ex:nm1}, which cannot solved by $Z3$ within $60$ min. Thus we encode them into a more powerful tool REDLOG for nonlinear constraints~\cite{DS:97:redlog}.
For every experiment in the table, the time to construct the HMM model and SMT queries is less than 1 second;
the time for solving SMT queries are less than 2 seconds, except for Ex.~\ref{ex:nm1}.

Among the mechanisms in Table~\ref{table:experiment-results}, Algorithm~\ref{algorithm:noisy-max},\ref{algorithm:above-threshold}
need our further investigation. We examine these algorithms carefully in the following subsections.


\vspace{-5pt}
\subsection{Noisy Max}
\label{subsection:noisymax}

Noisy Max is a simple yet useful data publishing mechanism in
differential privacy~\cite{DR:14:AFDP,DWWZK:18:DVDP}. Consider $n$
queries of the same range, say, the number of patients for $n$
different diseases in a hospital. We are interested in knowing which
of the $n$ diseases has the maximal number of patients in the hospital.
A simple privacy-respecting way to release the information is to add
independent noises to every query result and then return the index of
the maximal noisy results.

\vspace{-10pt}
\begin{small}
\begin{algorithm}
  \begin{algorithmic}[1]
    \Require{$0 \leq v_1, v_2, \ldots, v_n \leq 2$}
    \Ensure{The index $r$ with the maximal $\tilde{v}_r$ among $\tilde{v}_1, \tilde{v}_2, \ldots, \tilde{v}_n$}
    \Function{DiscreteNoisyMax}{$v_1, v_2, \ldots, v_n$}
      \State{$M, r, c \leftarrow -1, 0, 0$}
      \For{each $v_i$}
        \Match{$v_i$}
        \Comment{apply $\frac{1}{2}$-geometric mechanism}
          \lCase{$0$}
                {$\tilde{v}_i \leftarrow 0, 1, 2$ with probability
                 $\frac{2}{3}, \frac{1}{6}, \frac{1}{6}$}
          \lCase{$1$}
                {$\tilde{v}_i \leftarrow 0, 1, 2$ with probability
                 $\frac{1}{3}, \frac{1}{3}, \frac{1}{3}$}
          \lCase{$2$}
                {$\tilde{v}_i \leftarrow 0, 1, 2$ with probability
                 $\frac{1}{6}, \frac{1}{6}, \frac{2}{3}$}
        \EndMatch
        \If{$M = \tilde{v}_i$}
        \State{$c \leftarrow c + 1$}
        \State{$r \leftarrow i$ with probability $\frac{1}{c}$}
        \EndIf
        \If{$M < \tilde{v}_i$}
        \State{$M, r, c \leftarrow \tilde{v}_i, i, 1$}
        \EndIf
      \EndFor
      \State{\Return {$r$} }
    \EndFunction
  \end{algorithmic}
  \caption{Discrete Noisy Max}
  \label{algorithm:noisy-max}
\end{algorithm}
\end{small}
\vspace{-10pt}

In~\cite{DR:14:AFDP}, Noisy Max algorithm adds
continuous Laplacian noises to each query result.
The continuous Noisy Max algorithm is proved to effectively protect
privacy for neighboring data sets~\cite{DWWZK:18:DVDP}.
In practice continuous noises however are replaced by discrete noises using
floating-point numbers. Technically, the distribution of discrete
floating-point noises is different from the continuous distribution in
mathematics. Differential privacy can be breached~\cite{M:12:SLSBDP}.
The proof for continuous Noisy Max
algorithm does not immediately apply. Indeed, care must be taken to
avoid privacy breach.

\begin{wrapfigure}[10]{r}{8cm}
\vspace{-20pt}
  \centering
    \resizebox{0.6\columnwidth}{!}{
    \begin{tikzpicture}[->,>=stealth',shorten >=1pt,auto,node
      distance=2cm,node/.style={circle,draw}]
      \node at (-2.3, 1) { $\cdots$ };
      \node[node,scale=.67,label={[shift={(-.35,-.1)},scale=.67]$\blank$}]
           (v011) at (-1.6, 1) { $011$ };
      \node at ( -.8, 1) { $\cdots$ };
      \node[node,scale=.67,label={[shift={(-.35,-.1)},scale=.67]$\blank$}]
           (v120) at (   0, 1) { $120$ };
      \node at ( 1.0, 1) { $\cdots$ };
      \node[node,scale=.67,label={[shift={(-.35,-.1)},scale=.67]$\blank$}]
           (v202) at ( 2.0, 1) { $202$ };
      \node at ( 2.8, 1) { $\cdots$ };

      \node at (-1.8, -1) { $\cdots$ };
      \node[node, scale=0.67,
      label={[shift={(1,-0.5)},scale=0.67]
        \begin{tabular}{cl}
          $\tilde{2}(\frac{1}{2}),\tilde{3}(\frac{1}{2})$
                      & 
         \end{tabular}
          }]
          (022) at (0, -1) { $0\underline{2}\underline{2}$
      };
      \node at ( 2, -1) { $\cdots$ };

      \path
      (v011) edge [left,below] node [rotate=310,scale=.67]
      { $\frac{2}{3} \cdot \frac{1}{3} \cdot \frac{1}{3}$ }
      (022)

      (v120) edge [right,above] node [rotate=270,scale=.67]
      { $\frac{1}{3} \cdot \frac{2}{3} \cdot \frac{1}{6}$ }
      (022)

      (v202) edge [right,below] node [rotate=45,scale=.67]
      { $\frac{1}{6} \cdot \frac{1}{6} \cdot \frac{2}{3}$ }
      (022)
      ;
    \end{tikzpicture}
    }
    \vspace{-7pt}
  \caption{Hidden Markov Model for Noisy Max}
  \label{figure:hmm-noisy-max}
\end{wrapfigure}
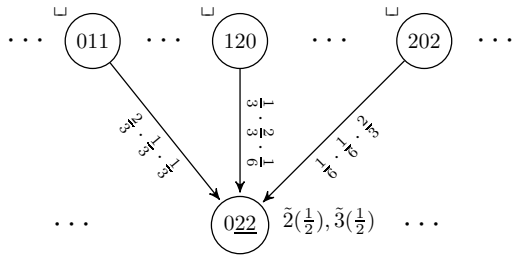

We introduce our algorithm and model. The standard algorithm is modified by adding discrete noises to
query results (Algorithm~\ref{algorithm:noisy-max}). In the
algorithm, the variables $M$ and $r$ contain the maximal noisy result
and its index respectively. We apply the truncated $\frac{1}{2}$-geometric
mechanism to each query with
the corresponding discrete range.
To avoid returning a fixed index when there are multiple noisy results with the same value,
the discrete algorithm
explicitly returns the index of the maximal noisy value with an equal probability (Line. $8-14$).

The HMM model with $n = 3$ queries is illustrated in Fig.~\ref{figure:hmm-noisy-max}.
The top states labeled $011$ and
$120$ correspond to three query results (on neighboring data sets) and $\blank$, i.e. nothing, is observed in the initial states.
Both states have a transition to the state $0\underline{2}\underline{2}$,
representing the perturbed query results obtained with different probabilities.
The index of the maximal result will be observed, which is $2$ or $3$ with probability $\frac{1}{2}$. Next we analyze Algorithm~\ref{algorithm:noisy-max} under the Pufferfish framework.

\hide{
Using our tool,
Algorithm~\ref{algorithm:noisy-max} doesn't satisfy $\ln(2)$-differentially private when $n = 3$.
We test a few values of $\epsilon$, to find out that the differential privacy budget $\epsilon$
for Algorithm~\ref{algorithm:noisy-max} is $1.3$
which is $2.1$.
We can say Algorithm~\ref{algorithm:noisy-max} is an improvement of Algorithm~\ref{algorithm:noisy-max-naive},
but still with a privacy budget larger than $\ln(2)$.
Thus it is not sufficient to replace continuous noises with
discrete noises. Additional steps are required to preserve privacy
and the privacy guarantee should be examined carefully.
An inexperienced mechanism designer may na\"ively discretize
the standard Noisy Max algorithm and put privacy at risk.
Next we analyze Algorithm~\ref{algorithm:noisy-max}
under the Pufferfish framework.
}

\begin{example}
\label{ex:nm1}
Consider three counting queries $f_A$, $f_B$, and $f_C$ for the number
of individuals contracting the diseases $A$, $B$, and $C$ respectively
in the data set $\calX^2$ with $\calX = \{ (0, 0, 0), (0, 0, 1),
\ldots, $(1$, $1,$ $1)$ \}$. An element $(a, b, c) \in \calX$ denotes
whether the data entry contracts the diseases $A$, $B$, and $C$
respectively. Assume that the contraction of each disease is independent
among individuals and the probabilities of contracting the diseases
$A$, $B$, and $C$ are $p_A$, $p_B$, and $p_C$ respectively.
The prior knowledge induces an information state for the
model in Fig.~\ref{figure:hmm-noisy-max}. For example, the state
$120$ has the probability $2p_A(1-p_A) \cdot p_B^2 \cdot (1-p_C)^2$.

Suppose John is in the data set and whether John contracts the disease
$A$ is a secret. We would like to check if the discrete Noisy Max
algorithm can protect the secret using the Pufferfish privacy
framework. Let us compute the initial information state $\pi$ given that
John has not contracted disease $A$. For instance, the initial
probability of the state $120$ is
$\frac{2p_A(1-p_A)}{(1-p_A)^2 + 2p_A(1-p_A)} \cdot p_B^2 \cdot
(1-p_C)^2$. 
The initial information state $\pi$ is obtained by computing the
probabilities of each of the $3^3$ top states.
Given that John has the disease $A$, the initial information state
$\tau$ is
computed similarly. In this case, the initial probability of the state
$120$ becomes
$\frac{2p_A(1-p_A)}{2p_A(1-p_A) + p_A^2} \cdot p_B^2 \cdot
(1-p_C)^2$. Probabilities of the $3^3$ top states form the initial
information state $\tau$.
From the initial information state $\pi$ and $\tau$, we compute the
probabilities of observing $\blank\tilde{1}$, $\blank\tilde{2}$, and
$\blank\tilde{3}$ in the formal model
(Fig.~\ref{figure:hmm-noisy-max}). The formulae for observation
probabilities are easy to compute.
However, the SMT solver Z3 cannot solve the non-linear formulae
generated by our algorithm.
In order to establish Pufferfish privacy automatically, we submit the
non-linear formulae to the constraint solver REDLOG. This time, the solver
successfully proves the HMM satisfying
$\ln(2)$-Pufferfish privacy.
\hide{
Instead, after we compute the probabilities of every observation,
we compare them pair by pair, construct formulas according to the condition of {$\epsilon$}-Pufferfish privacy
and use REDLOG to simplify the formulas by quantifier elimination. Finally, REDLOG returns FALSE for the formulas
we build, which verifies the problem as we want.
}
$\hfill\blacksquare$
\end{example}

Algorithm~\ref{algorithm:noisy-max} is {$\ln(2)$}-Pufferfish
private when the contraction of diseases is independent for data
entries.
\hide{
 which is inspired by when formulating differential privacy
in Pufferfish framework, independence among individuals in the database
is required~\cite{KM:14:PFMPD}.}
Our next step is to analyze the privacy mechanism model when
the contraction of the disease $A$ is correlated among data
entries.


\vspace{-5pt}
\begin{example}
\label{ex:nm2}
Assume that the data set consists of $2$ family members, including John, and there
are $5$ queries which ask the number of patients of $5$ diseases in the data set.
To protect privacy, Algorithm~\ref{algorithm:noisy-max} is applied to query results.
Now assume an attacker has certain prior knowledge:
1. Disease $1$ is so highly contagious that either none or both members infect the disease; 2. Disease $2$ to Disease $5$ are such diseases that every person
has the probability of $p_k$ to catch Disease $k$; and 3. The attacker knows the values of probabilities: $p_k = \frac{k}{10}$ for
$k \in \{3, 4, 5\}$, but does not know the value of $p_2$.
Suppose the secret is whether John has contracted Disease $1$ and we wonder whether there exists such a $p_2$
that $\ln(2)$-Pufferfish private is violated. We can compute the initial distribution pair
$\pi$ and $\tau$ given the above information.
For instance, if John has contracted Disease $1$, then the
initial probability for state $21110$ is 
$p_2(1-p_2) \cdot  (\frac{3}{10}) (1-\frac{3}{10}) \cdot
(\frac{4}{10}) (1-\frac{4}{10}) (1-\frac{5}{10})^2$. Similarly, we obtain the initial information state
given that John has not contracted the disease.
Then FAIER verifies the mechanism does not satisfy $\ln(2)$-Pufferfish private with the synthesized parameter $p_2 = \frac{1}{2}$.$\hfill\blacksquare$
\end{example}

\vspace{-5pt}
Provably correct privacy mechanisms can leak private information by seemingly harmless
modification or assumed prior knowledge. Ideally, privacy guarantees of practical mechanisms need
be re-established. 
Our verification
tool can reveal ill-designed privacy protection mechanisms
easily.

\subsection{Above Threshold}
\label{subsection:above-threshold}

Above threshold is a classical differentially private mechanism
for releasing numerical information~\cite{DR:14:AFDP}. Consider a data
set and an \textit{infinite} sequence of counting queries $f_1, f_2, \ldots$.
We would like to know the index of the first
query whose result is above a given threshold. In order to protect
privacy, the classical algorithm adds continuous noises on the
threshold and each query result. If the noisy query result is less
than the noisy threshold, the algorithm reports $\bot$ and continues
to the next counting query. Otherwise, the algorithm reports $\top$
and stops.

We consider counting queries with range $\{ 0, 1, 2 \}$ and apply the
truncated geometric mechanism for discrete noises. The discrete
above threshold algorithm is shown in
Algorithm~\ref{algorithm:above-threshold}. The algorithm first obtains
the noisy threshold $\tilde{t}$ using the truncated
$\frac{1}{4}$-geometric mechanism. For each query result $r_i$, it
computes a noisy result $\tilde{r}_i$ by applying the truncated
$\frac{1}{2}$-geometric mechanism. If $\tilde{r}_i < \tilde{t}$,
the algorithm outputs $\bot$ and continues. Otherwise, it halts with
the output $\top$.

\vspace{-15pt}
\subsubsection{Algorithm and Model}
To ensure $\epsilon$-differential privacy, the classical algorithm
applies the $\frac{2}{\epsilon}$- and $\frac{4}{\epsilon}$-Laplace
mechanism to the threshold and each query result respectively. The
continuous noisy threshold and query results are hence
$\frac{\epsilon}{2}$- and $\frac{\epsilon}{4}$-differentially private.
In Algorithm~\ref{algorithm:above-threshold}, the discrete noisy
threshold and query results are ${2\ln(2)}$- and
${\ln(2)}$-differentially private. If the classical proof still
applies, we expect the discrete above threshold algorithm is
${4\ln(2)}$-differentially private for $\frac{\epsilon}{2} = 2\ln(2)$.

Fig.~\ref{figure:hmm-above-threshold} gives an HMM
for Algorithm~\ref{algorithm:above-threshold}. In the model, the
state $t_ir_j$ represents the input threshold $t = i$ and the first
query result $r = f_1 (\od) = j$ for an input data set $\od$. From the
state $t_ir_j$, we apply the truncated $\frac{1}{4}$-geometric
mechanism. The state $\tilde{t}_ir_j$ hence means the noisy threshold
$\tilde{t} = i$ with the query result $r = j$. For instance, the state
$t_0r_1$ transits to $\tilde{t}_1r_1$ with probability $\frac{3}{20}$.
After the noisy threshold is obtained, we compute a noisy query result
by the truncated $\frac{1}{2}$-geometric mechanism. The state
$\tilde{t}_i\tilde{r}_j$ represents the noisy threshold $\tilde{t} =
i$ and the noisy query result $\tilde{r} = j$. In the figure, we
see that the state $\tilde{t}_1r_0$ moves to $\tilde{t}_1\tilde{r}_0$
with probability $\frac{2}{3}$. At the state $\tilde{t}_i\tilde{r}_j$,
$\top$ is observed if $j \geq i$; otherwise, $\bot$ is observed. From
the state $\tilde{t}_i\tilde{r}_j$, the model transits to the states
$\tilde{t}_ir_0$, $\tilde{t}_ir_1$, $\tilde{t}_ir_2$ with uniform
distribution. This simulates the next query result in
Algorithm~\ref{algorithm:above-threshold}. The model then continues to
process the next query.

\begin{algorithm}
  \begin{algorithmic}[1]
    \Procedure{AboveThreshold}{$\od$, $\{ f_1, f_2, \ldots \}$, $t$}
      \Match{$t$}
        \Comment{apply $\frac{1}{4}$-geometric mechanism}
        \lCase{$0$}{$\tilde{t} \leftarrow 0,1,2$ with probability
          $\frac{4}{5}, \frac{3}{20}, \frac{1}{20}$}
        \lCase{$1$}{$\tilde{t} \leftarrow 0,1,2$ with probability
          $\frac{1}{5}, \frac{3}{5}, \frac{1}{5}$}
        \lCase{$2$}{$\tilde{t} \leftarrow 0,1,2$ with probability
          $\frac{1}{20}, \frac{3}{20}, \frac{4}{5}$}
      \EndMatch
      \For{each query $f_i$}
        \State{$r_i \leftarrow f_i (\od)$}
        \Match{ $r_i$}
          \Comment{apply $\frac{1}{2}$-geometric mechanism}
          \lCase{$0$}{$\tilde{r}_{i} \leftarrow 0,1,2$ with probability
            $\frac{2}{3}, \frac{1}{6}, \frac{1}{6}$}
          \lCase{$1$}{$\tilde{r}_{i} \leftarrow 0,1,2$ with probability
            $\frac{1}{3}, \frac{1}{3}, \frac{1}{3}$}
          \lCase{$2$}{$\tilde{r}_{i} \leftarrow 0,1,2$ with probability
            $\frac{1}{6}, \frac{1}{6}, \frac{2}{3}$}
        \EndMatch
        \State{\textbf{if} $\tilde{r}_i \geq \tilde{t}$ \textbf{then}
          \textbf{halt} with $a_i = \top$ \textbf{else} $a_i = \bot$}
      \EndFor
    \EndProcedure
  \end{algorithmic}
  \caption{Input: private database $\od$,
    counting queries $f_i : \od \rightarrow \{ 0, 1, 2 \}$,
    threshold $t \in \{ 0, 1, 2 \}$; Output: $a_1, a_2, \ldots$}
  \label{algorithm:above-threshold}
\end{algorithm}

The bottom half of Fig.~\ref{figure:hmm-above-threshold} is another
copy of the model. All states in the second copy are
$\underline{\textmd{underlined}}$. For instance, the state
$\tilde{\underline{t}}_2\underline{r}_0$ represents the noisy
threshold is $2$ and the query result is $0$. Given an observation
sequence, the two copies are used to simulate the mechanism
conditioned on the prior knowledge with the two secrets. In the figure, we define the observation set $\Omega = \{
\blank, \bot, \top, 00, 01, 10, 11, 12, 21, 22$, $\spadesuit$,
$\heartsuit$, $\diamondsuit$, $\clubsuit \}$. At initial states
$t_ir_j$ and $\underline{t}_i\underline{r}_j$, only $\blank$ can be
observed. When the noisy threshold is greater than the noisy query result
($\tilde{t}_i\tilde{r}_j$ and
$\tilde{\underline{t}}_i\tilde{\underline{r}}_j$ with $i > j$), $\bot$
is observed. Otherwise, $\top$ is observed at states
$\tilde{t}_i\tilde{r}_j$ and
$\tilde{\underline{t}}_i\tilde{\underline{r}}_j$ with $i \leq j$.
Other observations are used to ``synchronize'' query results for neighboring data sets. More details are explained in Appendix 2.
\vspace{-20pt}

\subsubsection{Differential Privacy Analysis}
We can now perform differential privacy analysis using the HMM in Fig.~\ref{figure:hmm-above-threshold}. By
construction, each observation corresponds to a sequence of queries on
neighboring data sets and their results. If the proof of continuous
above threshold mechanism could carry over to our discretized
mechanism, we would expect differences of observation probabilities
from neighboring data sets to be bounded by the multiplicative factor
of $e^{4 \ln (2)} = 16$. Surprisingly, our tool always reports
larger differences as the number of queries increases. After
generalizing finite observations found by $Z3$, we obtain an
observation sequence of an arbitrary length described below.

Fix $n > 0$.
Consider a data set $\od$ such that $f_i (\od) = 1$ for $1 \leq i \leq
n$ and $f_{n+1} (\od) = 2$. A neighbor $\od'$ of $\od$ may
have $f_i (\od') = 2$ for $1 \leq i \leq n$ and $f_{n+1} (\od') = 1$.
Note that $|f_i (\od) - f_i (\od')| \leq 1$ for $1 \leq i
\leq n + 1$. $f_i$'s are counting queries. Suppose the threshold $t =
2$. Let us compute
the probabilities of observing $\bot^n\top$ on $\od$ and $\od'$.

\begin{figure}
  \centering
    \resizebox{0.75\columnwidth}{!}{
    \begin{tikzpicture}[->,>=stealth',shorten >=1pt,auto,node
      distance=2cm,node/.style={circle,draw,inner sep=0pt,minimum size=25pt}]
      \node[node, label={[shift={(-.6,-.5)}]$\blank$}]
            (t0r1) at (0, 4.65) { $t_0r_1$ };

      \draw (3.2, 3.95) rectangle (5.2, .05);
      \node at (4.95, 2.9) { $\cdots $};
      \node[node]
            (t'0r1) at (4.2, 2.9) { $\tilde{t}_0r_1$ };
      \node at (3.5, 2.9) { $\cdots $};
      \draw (-3.1, 3.95) rectangle (3.1, .05);
      \node[node, label={[shift={(0, 0)}, scale=.75]$00, 01, \heartsuit$}]
            (t'1r0) at (2.1, 2.9) { $\tilde{t}_1r_0$ };
      \node[node, label={[shift={(.6, 0)}, scale=.75]$10, 11, 12$}]
            (t'1r1) at (0, 2.9) { $\tilde{t}_1r_1$ };
      \node[node, label={[shift={(0, 0)}, scale=.75]$\spadesuit, 21, 22$}]
            (t'1r2) at (-2.1, 2.9) { $\tilde{t}_1r_2$ };
      \node at (0, 2.0) { $\vdots$ };
      \node at (-2.1, 2.0) { $\vdots$ };
      \draw (-3.2, 3.95) rectangle (-5.2, .05);
      \node at (-3.45, 2.9) { $\cdots$ };
      \node[node] (t'2r1) at (-4.2, 2.9) { $\tilde{t}_2r_1$ };
      \node at (-4.9, 2.9) { $\cdots$ };

      \node[node, label={[shift={(-.4, -1.3)}]$\bot$}]
            (t'1r'0) at (2.1, .9) { $\tilde{t}_1\tilde{r}_0$ };
      \node[node, label={[shift={(-.4, -1.3)}]$\top$}]
            (t'1r'1) at (0, .9) { $\tilde{t}_1\tilde{r}_1$ };
      \node[node, label={[shift={( .4, -1.3)}]$\top$}]
            (t'1r'2) at (-2.1, .9) { $\tilde{t}_1\tilde{r}_2$
      };

      \node[node, label={[shift={(-0.6,-0.7)}]$\blank$}]
            (tt0rr0) at (2.4, -4.65) { $\underline{t}_0\underline{r}_0$ };

      \draw (3.2, -3.95) rectangle (5.2, -.05);
      \node[node] (tt'0rr0) at (4.5, -2.9) { $\tilde{\underline{t}}_0\underline{r}_0$ };
      \node at (3.65, -2.9) { $\cdots$ };
      \draw (1.1, -3.95) rectangle (3.1, -.05);
      \node[node] (tt'1rr0) at (2.4, -2.9) { $\tilde{\underline{t}}_1\underline{r}_0$ };
      \node at (1.55, -2.9) { $\cdots$ };

      \draw (1., -3.95) rectangle (-5.2, -.05);
      \node[node, label={[shift={(0, -1.5)}, scale=.75]$00, 10, \clubsuit$}]
            (tt'2rr0) at (-0., -2.9)
            { $\tilde{\underline{t}}_2\underline{r}_0$ };
      \node[node, label={[shift={(0, -1.5)}, scale=.75]$01, 11, 21$}]
            (tt'2rr1) at (-2.1, -2.9)
            { $\tilde{\underline{t}}_2\underline{r}_1$ };
      \node[node, label={[shift={(0, -1.5)}, scale=.75]$\diamondsuit, 12, 22$}]
            (tt'2rr2) at (-4.2, -2.9)
            { $\tilde{\underline{t}}_2\underline{r}_2$ };

      \node[node, label={[shift={(-.4, -.1)}]$\bot$}]
            (tt'2rr'0) at (-.0, -.9)
            { $\tilde{\underline{t}}_2\tilde{\underline{r}}_0$ };

      \node at (-2.1, -1.75) { $\vdots$ };

      \node[node, label={[shift={(-.4, -.1)}]$\bot$}]
            (tt'2rr'1) at (-2.1, -.9)
            { $\tilde{\underline{t}}_2\tilde{\underline{r}}_1$ };
      \node at (-4.2, -1.75) { $\vdots$ };

      \node[node, label={[shift={(-.4, -.2)}]$\top$}]
            (tt'2rr'2) at (-4.2, -.8)
            { $\tilde{\underline{t}}_2\tilde{\underline{r}}_2$ };

      \path
      (t0r1) edge [bend left=22.5] node [above] { $\frac{4}{5}$ } (t'0r1)
      (t0r1) edge node [left] { $\frac{3}{20}$ } (t'1r1)
      (t0r1) edge [bend right=22.5] node [above] { $\frac{1}{20}$ } (t'2r1)

      (t'1r0) edge [bend right, very thick] node [right] { } (t'1r'0)
      (t'1r0) edge [ultra thin] node [above] { } (t'1r'1)
      (t'1r0) edge [ultra thin] node [above] { } (t'1r'2)

      (t'1r'0) edge [bend right] node {  } (t'1r0)
      (t'1r'0) edge node {  } (t'1r1)
      (t'1r'0) edge node {  } (t'1r2)

      (tt0rr0) edge [bend right=22.5] node [below] { $\frac{4}{5}$ } (tt'0rr0)
      (tt0rr0) edge node [left] { $\frac{3}{20}$ } (tt'1rr0)
      (tt0rr0) edge [bend left=22.5] node [below] { $\frac{1}{20}$ } (tt'2rr0)

      (tt'2rr0) edge [bend right, very thick] node [right] { } (tt'2rr'0)
      (tt'2rr0) edge [ultra thin] node [above] { } (tt'2rr'1)
      (tt'2rr0) edge [ultra thin] node [above] { } (tt'2rr'2)

      (tt'2rr'0) edge [bend right] node { } (tt'2rr0)
      (tt'2rr'0) edge node { } (tt'2rr1)
      (tt'2rr'0) edge node { } (tt'2rr2)
      ;
      \end{tikzpicture}
    }
  \caption{Hidden Markov Model for Above Threshold}
  \label{figure:hmm-above-threshold}
\end{figure}
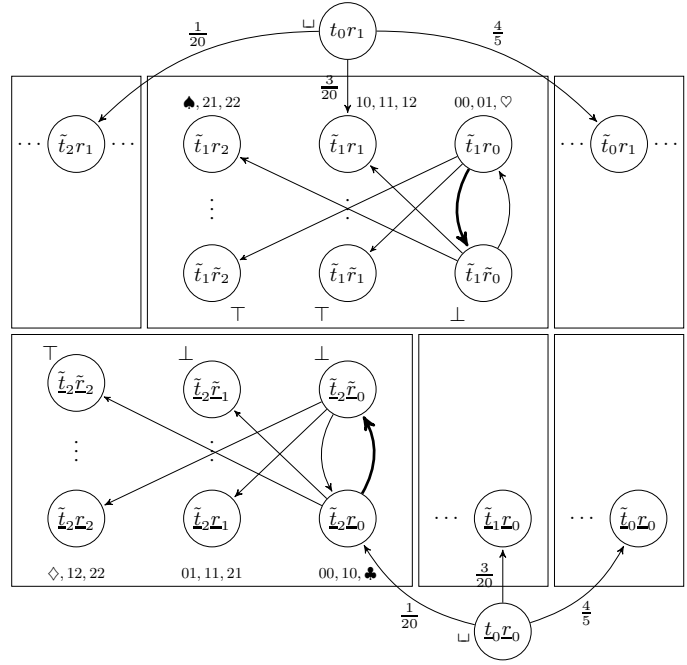

If $\tilde{t} = 0$, $\tilde{f}_1 \geq \tilde{t}$. The algorithm
reports $\top$ and stops. We cannot observe $\bot^n\top$: recall the assumption that $n > 0$.
It suffices to consider $\tilde{t} = 1$ or $2$. When $\tilde{t} = 1$, $\tilde{f}_i
(\od) = 0$ for $1 \leq i \leq n$ and $\tilde{f}_{n+1} (\od) \geq 1$.
Recall $f_i (\od) = 1$ for $1 \leq i \leq n$ and $f_{n+1} (\od) = 2$.
The probability of observing $\bot^n\top$ is
$(\frac{1}{3})^n \cdot \frac{5}{6}$. When $\tilde{t} = 2$, $\tilde{f}_1
(\od) \leq 1$ for $1 \leq i \leq n$ and $\tilde{f}_{n+1} (\od) = 2$. The
probability of observing $\bot^n\top$ is thus
$(\frac{2}{3})^n \cdot \frac{2}{3}$. In summary, the probability of
observing $\bot^n\top$ with $\od$ when $t = 2$ is
$\frac{3}{20} \cdot (\frac{1}{3})^n \cdot \frac{5}{6} + \frac{4}{5}
\cdot (\frac{2}{3})^n \cdot \frac{2}{3}$.
The case for $\od'$ is similar. When $\tilde{t} = 1$, the probability
of observing $\bot^n\top$ is $(\frac{1}{6})^n \cdot
\frac{2}{3}$. When $\tilde{t} = 2$, the probability of observing the
same sequence is $(\frac{1}{3})^n \cdot
\frac{1}{3}$. Hence the probability of observing $\bot^n\top$ with
$\od'$ when $t = 2$ is $\frac{3}{20} \cdot (\frac{1}{6})^n \cdot
\frac{2}{3} + \frac{4}{5} \cdot (\frac{1}{3})^n \cdot \frac{1}{3}$. Now,

\begin{align*}
    \frac{ \Pr(\omega = \bot^n\top | \od, t = 2) }
       { \Pr(\omega = \bot^n\top | \od', t = 2) }
  &=
  \frac{
    \frac{3}{20} \cdot (\frac{1}{3})^n \cdot \frac{5}{6} + \frac{4}{5}
    \cdot (\frac{2}{3})^n \cdot \frac{2}{3} }
       {
    \frac{3}{20} \cdot (\frac{1}{6})^n \cdot \frac{2}{3} + \frac{4}{5}
    \cdot (\frac{1}{3})^n \cdot \frac{1}{3} }\\
  &>
  \frac{ \frac{4}{5} \cdot (\frac{2}{3})^n \cdot \frac{2}{3} }
       { \frac{3}{20} \cdot (\frac{1}{3})^n \cdot \frac{2}{3} +
         \frac{4}{5} \cdot (\frac{1}{3})^n \cdot \frac{1}{3} }
  = \frac{ \frac{8}{15} (\frac{2}{3})^n }
         { \frac{11}{30} (\frac{1}{3})^n }
  = \frac{16}{11} \cdot 2^n.
\end{align*}

\vspace{-5pt}
We see that the ratio of $\Pr (\omega = \bot^n\top | \od, t = 2)$ and
$\Pr (\omega = \bot^n\top | \od', t = 2)$ can be arbitrarily large.
Unexpectedly, the discrete above threshold
cannot be $\epsilon$-differentially private for any $\epsilon$.
Replacing continuous noises with truncated
discrete noises does not preserve any privacy at all.
This case emphasizes the importance of applying verification technique
to practical implementations.

\section{Combining Techniques for Differential Privacy}
\label{section:experiments}
In this section, we investigate into two state-of-the-art
tools for detecting violations of differential privacy, namely StatDP~\cite{DWWZK:18:DVDP} and DP-Sniper~\cite{BSBV:21:dpsniper}, to compare with our tool. We decide to choose these tools as baselines since they support programs with arbitrary loops and  
arbitrary sampling distributions.  On the contrary, DiPC~\cite{BCJSV:20:ddppfio,BCJ:19:DiPC}, DP-Finder~\cite{BTDPM:18:dpfinder} and CheckDP~\cite{WDKZ:20:CDP} et al. do not support arbitrary loops or only synthesize proofs for privacy budget $\epsilon$ when Laplace distributions are applied. In order to compare with our tool FAIER, the discrete mechanisms with truncated geometric distributions are implemented in
these tools. We present comparisons in Subsection~\ref{subsection:com}, and moreover, in Subsection~\ref{sub:combine}, we discuss how testing
and our verification technique can be combined to certify counterexamples and find the precise lower bound for privacy budget.
\vspace{-5pt}

\subsection{Comparison}
\label{subsection:com}
\subsubsection{Different problem statements}
As all the tools can be used to find the privacy budget $\epsilon$ for differential private mechanisms, the problem statements they address are different: I. With a fixed value of $\epsilon$, StatDP runs the mechanism repeatedly and tries to report the output event that makes the mechanism violate $\epsilon$-differential privacy, with a p-value as the confidence level. If the p-value is below 0.05, StatDP is of high confidence that $\epsilon$-differential privacy is violated; Otherwise the mechanism is very likely (depending on the p-value) to satisfy. II. On the other hand, DP-Sniper aims to learn for the optimal attack that maximizes the ratio of probabilities for certain outputs on all the neighboring inputs. Therefore it returns the corresponding ``optimal'' witness (neighboring inputs) along with a value $\epsilon$ such that the counterexample violates $\epsilon$-differential privacy with $\epsilon$ as large as possible. III. Differently, FAIER makes use of the HMM model and examines all the pairs of neighboring inputs and outputs to make sure that $\epsilon$-differential privacy is satisfied by all cases, or violated by an counterexample, with a fixed value of $\epsilon$. IV. Note that FAIER is aimed at Pufferfish privacy verification where prior knowledge can affect the data sets distributions and unknown parameters are allowed, which are not involved in the other tools. Meanwhile, the others support continuous noise while FAIER does not (unless an HMM with finite state space can be obtained).

\begin{wraptable}[12]{r}{7cm}
\vspace{-15pt}
\caption{Heuristic input patterns used in StatDP and DP-Sniper, from ~\cite{DWWZK:18:DVDP}}
  \label{table:InputPattern}
  \centering
  \vspace{-10pt}
  \[
    \begin{array}{c|cc}
      Category & D_1 & D_2 \\
      \hline
      \textmd{One Above}
      & [1,1,1,1,1]
      & [2,1,1,1,1]
      \\
            \textmd{One Below}
      & [1,1,1,1,1]
      & [0,1,1,1,1]
      \\
            \textmd{One Above Rest Below}
      & [1,1,1,1,1]
      & [2,0,0,0,0]
      \\
            \textmd{One Below Rest Above}
      & [1,1,1,1,1]
      & [0,2,2,2,2]
      \\
            \textmd{Half Half}
      & [1,1,1,1,1]
      & [0,0,2,2,2]
      \\
            \textmd{All Above \& All Below}
      & [1,1,1,1,1]
      & [2,2,2,2,2]
      \\
      \textmd{X Shape}
      & [1,1,0,0,0]
      & [0,0,1,1,1]\\
    \end{array}
  \]
\end{wraptable}

\vspace{-10pt}
\subsubsection{Efficiency and precision}
We make comparison of the tools in terms of efficiency and precision by performing experiments on Discrete Noisy Max (Algorithm.~\ref {algorithm:noisy-max}) with $n=5$ queries. The lower bound~\cite{BTDPM:18:dpfinder} of the privacy budget, i.e., 
the largest $\epsilon$ that the mechanism is not $\epsilon$-differential privacy, is $1.372$ up to a precision of 0.001.
I. Fix an $\epsilon$, StatDP takes $8$ seconds on average to report an event 0 along with a p-value under the usual setting of $100k/500k$ times for event selection/counterexample detection.
However, there is a need for specifying the range of $\epsilon$
in advance and more values of $\epsilon$ to test
will consume more time. We first select $\epsilon$ increasingly with a step of $0.1$ in the range of [0, 2]. Then the range is narrowed down according to the p-values and we select $\epsilon$ in the range with a smaller step $0.01$ and so on. The similar process also applies for FAIER. Altogether StatDP takes around $600$ seconds to get an overview of the results. Fast enough, though, it has the drawback of instability and the precision is lower than the other tools.
It reports the mechanism satisfies $1.344$-differential privacy in the first execution, which is incorrect, and reports it violates $1.353$-differential privacy in the second execution.

II. DP-Sniper returns a witness [0,2,2,2,2] and [1,1,1,1,1] with $\epsilon = 1.371$ for three times, which is correct, stable and the result is almost the true lower bound. However, it takes around $4600$ seconds on average to train a multi-layer perceptron with $10000k$ samples and get this result. Unlike the evaluation in~\cite{BSBV:21:dpsniper}, DP-Sniper performs much slower than StatDP when it comes to discrete random noise. The reason is that DP-Sniper cannot use high-efficient sampling commands such as numpy.random.laplace to get all the samples at once. It has to calculate and sample different distributions according to different inputs. We've tried to use numpy.random.choice to sample different distributions, but it is inefficient for small vectors and wouldn't terminate for more than 10 hours in our experiment. We've also tried to reduce the number of samples to $1000k$. This time it terminates with $308$ seconds with an imprecise $\epsilon = 1.350$.

III. FAIER takes less than 1 second to build the HMM model and 160 seconds to compute SMT query for every data set (possible initial state), which will be later used to compute on neighboring data sets if an $\epsilon$ is assigned. The results returned by FAIER are the most precise ones. It takes Z3 $523$ seconds to verify that $1.373$-differential privacy is satisfied and $234$ seconds that $1.372$-differential privacy is violated witnessed by the input pair [0,2,2,2,2] and [1,1,1,1,1] and output event 1. It takes only $40$ seconds to verify when $\epsilon= 1.34$, a little far away from the true lower bound. Altogether it takes around $1600$ seconds to assure the true bound, which is acceptable.

\vspace{-5pt}
\subsection{Combining Verification and Testing}\label{sub:combine}

The findings during experiments inspire us to combine verification (FAIER) and testing (DP-Sniper, StatDP) together to efficiently make use of each tool. First, we can see that the witnesses found by FAIER and DP-Sniper are the same one. 
Actually, if heuristic searching strategies for input pairs are used, i.e., Table.~\ref{table:InputPattern} used in DP-Sniper and StatDP, FAIER will quickly find the violation pairs, which saves huge time in the occasions of privacy violations. 
Second, since the witness returned by DP-Sniper is the optimal input pair that maximize the probability difference, FAIER can precisely verify whether the ``optimal'' witness satisfies $\epsilon$-differential privacy, whereby FAIER will more likely to find the true lower bound as $\epsilon$ increase in short time. Third, since StatDP returns an imprecise result quickly given an $\epsilon$, we can combine StatDP and FAIER to efficiently get a precise lower bound. The pseudo-code is in Algorithm~\ref{algorithm:lowerbound}.

Algorithm~\ref{algorithm:lowerbound} first feeds mechanism $M$ as input to the testing tool StatDP,
to obtain an interval $I$ whose left end point is $\epsilon$ with p-value $<0.05$
and right end point with p-value $=1$. StatDP can conclude if p-value$<0.05$, the mechanism doesn't satisfy $\epsilon$-differential privacy with high confidence
and if p-value$=1$, the mechanism satisfies for sure. However, for other p-values, StatDP
is not confident to give useful conclusions.
Here is where our tool can work out ---
FAIER can determine whether $M$ satisfies $\epsilon$-differential privacy, given any $\epsilon$.
As a result we can combine to efficiently get arbitrarily close to the lower bound $\epsilon$ wrt. a given precision by binary search.
For instance, we apply StatDP on Algorithm~\ref{algorithm:noisy-max}
to get an interval $I = [1.34, 1.38]$ according to the p-value graph, and then apply our tool FAIER to verify $\epsilon$-differential privacy.
Consequently, our tool reports the lower bound is $1.372$ (up to a precision of 0.001).

\vspace{-15pt}
\begin{footnotesize}
\begin{algorithm}
  \begin{algorithmic}[1]
    \Procedure{Compute lower bound}{Mechanism M}
      \State{Use StatDP with input M to get an interval I}
      \Comment{the left end point is an $\epsilon$ with p-value$<0.05$
      and the right one is one with p-value$=1$}
      \State{Apply binary search on I, in each iteration the value is $\epsilon$}
      \Repeat
      \State{Use FAIER with input M and $\epsilon$}
      \If{result is SAT}
      \Comment{not satisfy $\epsilon$-differential privacy}
          \State{left end point = $\epsilon$}
      \Else \Comment{satisfy $\epsilon$-differential privacy}
          \State{right end point = $\epsilon$}
      \EndIf
      \Until{reaching required precision}
      \State{\Return {$\epsilon$} }
    \EndProcedure
  \end{algorithmic}
  \caption{Pseudo-code to compute the lower bound}
  \label{algorithm:lowerbound}
\end{algorithm}
\end{footnotesize}

\section{Related Work}
\label{section:related-work}

\hide{
\paragraph{Pufferfish privacy} The definition of Pufferfish privacy allows correlations among data items~\cite{KM:14:PFMPD},
which is an elegant generalization of differential privacy~\cite{D:06:DP}.
Since then many follow-up works have explored the topic.
Song et al.~\cite{SWC:17:ppmcd} provided the first efficient mechanism to apply Pufferfish Privacy,
which is the generalization of continuous Laplace mechanism for differential privacy.
To the best of our knowledge, our algorithm is the first one that could automatically verify discrete Pufferfish privacy mechanisms.}
\vspace{-8pt}
\paragraph{Methods of proving/testing differential privacy.}
Barthe et al.~\cite{BKO:12:prrdp,BKO:12a:prrdp} proposed to prove differential privacy at the beginning. Then a number of work~\cite{BGGHS:16:PDPPC,BGHJP:16:PLTDP,AH:17:SCPDP} extended probabilistic relational Hoare logic
and applied approximate probabilistic couplings between programs on adjacent inputs.
They successfully proved differential privacy for several algorithms,
but cannot disprove
privacy. Zhang et al.~\cite{ZK:17:lightdp,WDW:19:shadowdp,WDKZ:20:CDP} proposed to apply randomness alignment to evaluate privacy cost
and implemented CheckDP that could rewrite classic privacy mechanisms involving Laplacian noise to verify differential privacy. 
Bichsel et.al~\cite{BTDPM:18:dpfinder},
Ding et.al~\cite{DWWZK:18:DVDP} and Zhang et.al~\cite{ZRH:20:DPCheck} used
testing and searching to find violations for differential privacy mechanisms, the results of which may be too coarse or imprecise. 
Liu et al.~\cite{LWZ:18:MCDPP} chose Markov chains and Markov decision processes to model deferentially private mechanisms and verify privacy properties in extended probabilistic temporal logic. McIver et al.~\cite{MM:19:PPADP} applied Quantitative Information Flow to analyze Randomized response mechanism in differential privacy.
We note that all the automated tools above for proving or testing differential privacy, plus ours, have not been well studied in privacy mechanisms with considerably large data sets.
\vspace{-8pt}
\paragraph{Complexity in verifying differential privacy.}

Gaboardi et al.~\cite{GNP:20:CVLPDP} studied the problem of verifying differential privacy for probabilistic loop-free programs.
They showed that to decide $\epsilon$-differential privacy is $\mathbf{coNP^{\#P}}$-complete
and to approximate the level of differential privacy is both $\mathbf{NP}$-hard and $\mathbf{coNP}$-hard.
Barthe et al.~\cite{BCJSV:20:ddppfio} first proved that checking differential privacy is undecidable.
The difference with our work lies in that we study verification problems for mechanisms modeled in HMMs in Pufferfish privacy. Chistikov et al.~\cite{CKMP:2020:LMC} proved that the big-$O$ problem for labeled Markov chains (LMCs) is undecidable, which is similar to deciding the ratio of two probabilities in differential privacy. Though, their proof does not apply here since HMMs in our paper do not have the same non-deterministic power as LMCs.


\hide{
\vspace{-15pt}
\subsubsection{Acknowledgements.}
We would like to thank the anonymous reviewers for their
valuable suggestions and comments about this paper. The work is
supported by Ministry of Science and Technology of Taiwan
under the Grant Number 108-2221-E-001-010-MY3; 
the Data Safety and Talent Cultivation Project
AS-KPQ-109-DSTCP.

\ldp{more projects to be added..}
}
\section*{Appendix 1}

\label{section:pc}
\begin{proof}
  In order to satisfy Pufferfish privacy in hidden Markov model, we have to decide whether
  expressions (\ref{max1}) and (\ref{max2}) are no more than $0$.
  Let's just simplify the problem by only having one initial distribution pair to compare
  so that we only need to find the observation sequence.
  We will show the problem to find the maximal value is NP-hard by
  a reduction from the classic Boolean Satisfiability Problem (SAT), which is known to be NP-hard. To be specific, given an arbitrary
  formula in conjuncted normal form, we construct a corresponding hidden Markov model under Pufferfish privacy framework,
  such that the formula is satisfiable if and only if the expressions (\ref{max1}) and (\ref{max2}) both take the maximal value $0$.



  Assume we have a formula $F(x_1,\ldots,x_n)$ in conjuncted normal form,
  with $n(n>=3)$ variables and $m$ clauses, $C_1,\ldots,C_m$. We shall construct a hidden Markov model $H = (K, \Obs, o)$
  such that with $\epsilon = \ln (4)$, expressions (\ref{max1}) and (\ref{max2}) will take maximal value $0$
  if and only if the formula $F(x_1,\ldots,x_n)$ is satisfiable.

  \paragraph{Construction.} The construction of model is similar to that in \cite{PCT:87:CMDP}. We first describe the Markov Chain $K =
  (S, p)$. $S$ contains a state group $A$ with six states $A_{ij}$, $A'_{ij}$, $T_{A ij}$, $T'_{Aij}$, $F_{Aij}$, $F'_{Aij}$ and
  a state group $B$ with six states $B_{ij}$, $B'_{ij}$, $T_{Bij}$, $T'_{Bij}$, $F_{Bij}$, $F'_{Bij}$
  for each clause $C_i$ and variable $x_j$. Besides, there are $4m$ states $A_{i,n+1}$, $A'_{i,n+1}$, $B_{i,n+1}$, $B'_{i,n+1}$ for each clause $C_i$.
  The transition distribution $p$ is as follows. For group $A$, there are two transitions with same probability $\frac{1}{2}$ leading from
  state $A_{ij}$ to $T_{Aij}$ and $F_{Aij}$ respectively; similarly there are two transitions leading with probability $\frac{1}{2}$
  from $A'_{ij}$ to $T'_{Aij}$ and $F'_{Aij}$. There's only one transition leading with certainty from $T_{Aij}$, $F_{Aij}$, $T'_{Aij}$, $F'_{Aij}$,
  to $A_{i,j+1}$, $A_{i,j+1}$, $A'_{i,j+1}$, $A'_{i,j+1}$ respectively with two exceptions: If $x_j$ appears positively in $C_i$,
  the transition from $T'_{Aij}$ is to $A_{i,j+1}$ instead of $A'_{i,j+1}$; and if $x_j$ appears negatively, the transition from
  $F'_{Aij}$ is to $A_{i,j+1}$. For the state group $B$, all the transitions imitate that in group $A$ only with different state names.
  For instance, there are two transitions leading with same probability $\frac{1}{2}$ from state $B_{ij}$ to $T_{Bij}$ and $F_{Bij}$ and so on.

  Next we describe the observations $\Obs$ and the observation distribution. In state
  $A_{ij}$, $A'_{ij}$, $B_{ij}$, $B'_{ij}$ with $1\leq j \leq n$, one can observe $X_j \in \Obs$ with certainty.
  In state $T_{Aij}$, $T'_{Aij}$, $T_{Bij}$, $T'_{Bij}$ with $1\leq j \leq n$, one can only observe $T_j \in \Obs$;
  similarly, the sole observation $F_j \in \Obs$ can be observed in state $F_{Aij}$, $F'_{Aij}$, $F_{Bij}$, $F'_{Bij}$ with $1\leq j \leq n$.
  In state $A_{i,n+1}$, we have probability $\frac{4}{5}$ to observe $\top \in \Obs$ and $\frac{1}{5}$  to observe $\bot \in \Obs$;
  while in state $B_{i,n+1}$, we have probability $\frac{1}{5}$ to observe $\top$ and $\frac{4}{5}$  to observe $\bot$.
  In state $A'_{i,n+1}$ and $B'_{i,n+1}$, there are equal probabilities of $\frac{1}{2}$ observing $\top$ and $\bot$.

  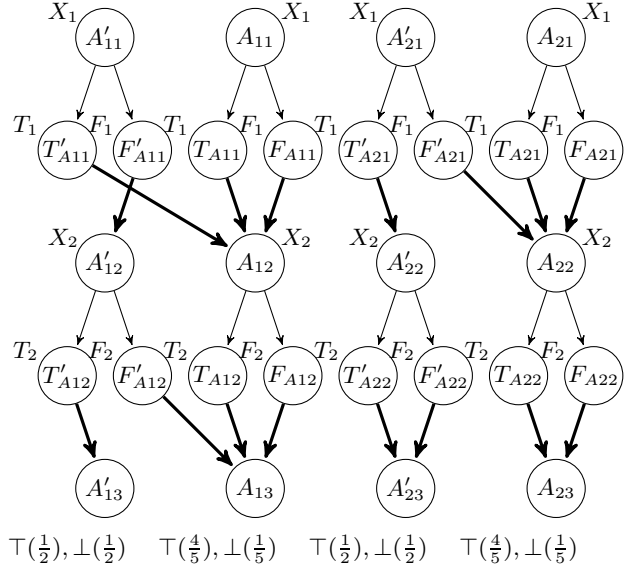
\begin{figure}
  \center
    \begin{tikzpicture}[xscale = 0.5, ->,>=stealth',shorten >=1pt,auto,node
      distance=1cm,
      node/.style={circle,draw,inner sep=0pt,minimum size=22pt}]

      \node[node, label={[shift={(-.55,-.3)}]$X_1$}]
            (AA11) at  ( -8,    3) { $A'_{11}$ };
      \node[node, label={[shift={(-.55,-.3)}]$T_1$}]
            (TTA11) at ( -9,  1.5) { $T'_{A11}$ };
      \node[node, label={[shift={(-.55,-.3)}]$F_1$}]
            (FFA11) at ( -7,  1.5) { $F'_{A11}$ };

      \node[node, label={[shift={(-.55,-.3)}]$X_2$}]
            (AA12) at  ( -8,    0) { $A'_{12}$ };
      \node[node, label={[shift={(-.55,-.3)}]$T_2$}]
            (TTA12) at ( -9, -1.5) { $T'_{A12}$ };
      \node[node, label={[shift={(-.55,-.3)}]$F_2$}]
            (FFA12) at ( -7, -1.5) { $F'_{A12}$ };

      \node[node,
            label={[shift={(-0.5,-1.5)}]$\top(\frac{1}{2}),\bot(\frac{1}{2})$}]
            (AA13) at  ( -8,   -3) { $A'_{13}$ };

      \node[node, label={[shift={( .55,-.3)}]$X_1$}]
            (A11) at  ( -4,    3) { $A_{11}$ };
      \node[node, label={[shift={( -.55,-.3)}]$T_1$}]
            (TA11) at ( -5,  1.5) { $T_{A11}$ };
      \node[node, label={[shift={( -.55,-.3)}]$F_1$}]
            (FA11) at ( -3,  1.5) { $F_{A11}$ };

      \node[node, label={[shift={( .55,-.3)}]$X_2$}]
            (A12) at  ( -4,    0) { $A_{12}$ };
      \node[node, label={[shift={( -.55,-.3)}]$T_2$}]
            (TA12) at ( -5, -1.5) { $T_{A12}$ };
      \node[node, label={[shift={( -.55,-.3)}]$F_2$}]
            (FA12) at ( -3, -1.5) { $F_{A12}$ };

      \node[node,
            label={[shift={(-0.5,-1.5)}]$\top(\frac{4}{5}),\bot(\frac{1}{5})$}]
            (A13) at ( -4,   -3) { $A_{13}$ };

      \node[node, label={[shift={(-.55,-.3)}]$X_1$}]
            (AA21) at  (  0,    3) { $A'_{21}$ };
      \node[node, label={[shift={(-.55,-.3)}]$T_1$}]
            (TTA21) at (  -1,  1.5) { $T'_{A21}$ };
      \node[node, label={[shift={(-.55,-.3)}]$F_1$}]
            (FFA21) at (  1,  1.5) { $F'_{A21}$ };

      \node[node, label={[shift={(-.55,-.3)}]$X_2$}]
            (AA22) at  (  0,    0) { $A'_{22}$ };
      \node[node, label={[shift={(-.55,-.3)}]$T_2$}]
            (TTA22) at (  -1, -1.5) { $T'_{A22}$ };
      \node[node, label={[shift={(-.55,-.3)}]$F_2$}]
            (FFA22) at (  1, -1.5) { $F'_{A22}$ };

      \node[node,
            label={[shift={( -0.5,-1.5)}]$\top(\frac{1}{2}),\bot(\frac{1}{2})$}]
            (AA23) at  (  0,   -3) { $A'_{23}$ };

      \node[node, label={[shift={( .55,-.3)}]$X_1$}]
            (A21) at  (  4,    3) { $A_{21}$ };
      \node[node, label={[shift={( -.55,-.3)}]$T_1$}]
            (TA21) at (  3,  1.5) { $T_{A21}$ };
      \node[node, label={[shift={( -.55,-.3)}]$F_1$}]
            (FA21) at (  5,  1.5) { $F_{A21}$ };

      \node[node, label={[shift={( .55,-.3)}]$X_2$}]
            (A22) at  (  4,    0) { $A_{22}$ };
      \node[node, label={[shift={( -.55,-.3)}]$T_2$}]
            (TA22) at (  3, -1.5) { $T_{A22}$ };
      \node[node, label={[shift={( -.55,-.3)}]$F_2$}]
            (FA22) at (  5, -1.5) { $F_{A22}$ };

      \node[node,
            label={[shift={( -0.5,-1.5)}]$\top(\frac{4}{5}),\bot(\frac{1}{5})$}]
            (A23) at (  4,   -3) { $A_{23}$ };

      \path
      (AA11)  edge (TTA11)
      (AA11)  edge (FFA11)
      (TTA11) edge[very thick] (A12)
      (FFA11) edge[very thick] (AA12)

      (AA12)  edge (TTA12)
      (AA12)  edge (FFA12)
      (TTA12) edge[very thick] (AA13)
      (FFA12) edge[very thick] (A13)

      (A11)  edge (TA11)
      (A11)  edge (FA11)
      (TA11) edge[very thick] (A12)
      (FA11) edge[very thick] (A12)

      (A12)  edge (TA12)
      (A12)  edge (FA12)
      (TA12) edge[very thick] (A13)
      (FA12) edge[very thick] (A13)

      (AA21)  edge (TTA21)
      (AA21)  edge (FFA21)
      (TTA21) edge[very thick] (AA22)
      (FFA21) edge[very thick] (A22)

      (AA22)  edge (TTA22)
      (AA22)  edge (FFA22)
      (TTA22) edge[very thick] (AA23)
      (FFA22) edge[very thick] (AA23)

      (A21)  edge (TA21)
      (A21)  edge (FA21)
      (TA21) edge[very thick] (A22)
      (FA21) edge[very thick] (A22)

      (A22)  edge (TA22)
      (A22)  edge (FA22)
      (TA22) edge[very thick] (A23)
      (FA22) edge[very thick] (A23)
      ;
      \end{tikzpicture}
    \caption{Construction for $(x_1 \vee \neg x_2) \wedge \neg x_1$}
    \label{figure:proof-construction}
  \end{figure}

  Fig.~\ref{figure:proof-construction} illustrates a part of the
  construction for the CNF formula $(x_1 \vee \neg x_2) \wedge \neg
  x_1$. State names are shown inside circles. Thin arrows represent
  transitions with probability $\frac{1}{2}$; thick arrows represent
  transitions with probability $1$. Observation distributions are
  shown outside each states. For instance, $X_1$ is observed with
  probability $1$ at the state $A'_{11}$. At the state $A'_{13}$,
  $\top$ and $\bot$ are observed with probability $\frac{1}{2}$
  each.

  In the figure, the left-hand side corresponds to the clause $x_1
  \vee \neg x_2$. Since the variable $x_1$ appears positively in the
  clause, there is a transition from $T'_{A11}$ to $A_{12}$ with
  probability $1$ according to the construction. Similarly, another
  transition from $F'_{A12}$ to $A_{13}$ with probability $1$ is
  needed for the negative occurrence of the variable $x_2$ in the
  clause. For the right-hand side corresponding to the clause $\neg
  x_1$, a transition from $F'_{A21}$ to $A_{22}$ with probability $1$
  is added.

  The construction for the state group $B$ is almost identical except
  the observation distributions on the states $B_{13}$ and $B_{23}$.
  At the states $B_{13}$ and $B_{23}$, $\top$ and $\bot$ can be
  observed with probabilities $\frac{1}{5}$ and $\frac{4}{5}$
  respectively. The construction for the state group $B$ is not shown
  in the figure for brevity.

  Then we describe the Pufferfish privacy scenario in this hidden Markov model. Assume that according to
  prior knowledge $\bbfD$ and discriminative secrets $\bbfS_{\textmd{pairs}}$, we only have
  one initial distribution pair $D_1$ and $D_2$ to compare. $D_1$ induces a uniform distribution,
  to start from each member in the state set $\{A'_{i1}\}$ with $1 \leq i \leq m$, whose probability is $\frac{1}{m}$.
  Similarly, in $D_2$, the probability starting from each member in the state set $\{B'_{i1}\}$ is also $\frac{1}{m}$ with $1 \leq i \leq m$.
  We set the parameter $\epsilon = \ln (4)$.

  \paragraph{Reduction.} The intuition is that starting from state $A'_{i1}$ or $B'_{i1}$,
  the clause $C_i$ is chosen and then the assignment of each variable will be considered one by one
  in this clause. Once the assignment of a variable $x_j$ makes $C_i$ satisfied, immediately
  state $A_{i,j+1}$ or $B_{i,j+1}$ is reached. So at last if state $A'_{i,n+1}$ or $B'_{i,n+1}$
  is reached, it means that the clause $C_i$ is not satisfied under this assignment. Now, we
  claim that  $\Pr (\calM (D_1) = \overline{\omega}) - 4 \times \Pr (\calM (D_2) = \overline{\omega})$ takes the maximal value $0$
  if and only if $\overline{\omega}$ is the observation sequence $X_1A_1X_2\ldots A_n \top$ such that
  formula $F(x_1,\ldots,x_n)$ is satisfied under assignment with $A_i \in \{T_i,F_i\}$ for each variable $x_i$
  (Similar analysis applies for $\Pr (\calM (D_2) = \overline{\omega}) - 4 \times \Pr (\calM (D_1) = \overline{\omega})$ except that it takes the
   maximal value $0$ with $\bot$ as the last observation).

  We argue that $0$ is the maximal value.
  It's easy to see that if we take an arbitrary observation sequence $\overline{\omega} = X_1A_1X_2\ldots $,
  as long as $\top$ or $\bot$ hasn't been observed, $\Pr (\calM (D_1) = \overline{\omega}) -  4 \times \Pr (\calM (D_2) = \overline{\omega}) < 0$.
  That's because the state group $B$ just imitate the state group $A$ before reaching the state
  $B_{i,n+1}$ and $B'_{i,n+1}$. Thus the maximal value must be less than 0 or be obtained after we observe $\top$
  or $\bot$.

  Then we consider $\overline{\omega}=X_1A_1X_2\ldots A_n \top$. Note that if $C_i$ is satisfied under observation $\overline{\omega}$, we start from $A'_{i1}$ and $B'_{i1}$ both with
  probability $\frac{1}{m}$, finally reaching $A_{i,n+1}$ and $B_{i,n+1}$ with probabilities $2^{-n} \times \frac{1}{m} \times \frac{4}{5}$
  and $2^{-n} \times \frac{1}{m} \times \frac{1}{5}$ respectively;
  if $C_i$ is not satisfied, we finally reach $A'_{i,n+1}$ and $B'_{i,n+1}$  with equal probabilities of $2^{-n}\times \frac{1}{m} \times \frac{1}{2}$.
  Thus a satisfied clause will contribute $2^{-n} \times \frac{1}{m} \times \frac{4}{5} - 4 \times 2^{-n} \times \frac{1}{m} \times \frac{1}{5} = 0$ to
  the result; while if some clause is not satisfied, $\Pr (\calM (D_1) = \overline{\omega})- 4\times \Pr (\calM (D_2) = \overline{\omega})$ is strictly less than $0$.
  Therefore, if we choose a observation sequence ended with $\top$ such that all the
  clauses are satisfied, $\Pr (\calM (D_1) = \overline{\omega})- 4\times \Pr (\calM (D_2) = \overline{\omega})$ will take the maximal value 0.
  If we consider $\overline{\omega}=X_1A_1X_2\ldots A_n \bot$, similar analysis concludes that $\Pr (\calM (D_1) = \overline{\omega}) - 4 \times \Pr (\calM (D_2) = \overline{\omega})$
  will be strictly less than $0$.
  This indicates that $0$ is the maximal value of $\Pr (\calM (D_1) = \overline{\omega}) - 4 \times \Pr (\calM (D_2) = \overline{\omega})$
  among all the observation sequences.

  Finally from the process above, it's easy to see that
  $\Pr (\calM (D_1) = \overline{\omega}) - 4 \times \Pr (\calM (D_2) = \overline{\omega})$ takes the maximal value $0$
  if and only if $F(x_1,\ldots,x_n)$ is satisfied under observation sequence $\overline{\omega}=X_1A_1X_2\ldots A_n \top$
  with assignment $A_i \in \{T_i,F_i\}$ for each variable $x_i$.
  Since determining whether Pufferfish privacy is preserved
  is equivalent to determining whether the maximal value is above $0$,
  we prove that the general problem for $\epsilon$-Pufferfish privacy is NP-hard.
\end{proof} 

\section*{Appendix 2}
\label{subsection:correctness}
In the literature, if the perturbed query result is smaller than the perturbed
threshold, noise will be added into next query, the result of which is uncertain.
Thus, nondeterminism is required here to choose the next query and ~\cite{LWZ:18:MCDPP} uses
a Markov decision process to model the algorithm. In order to model nondeterminism in an HMM,
we assign equal probabilities to return to all the possible queries. For instance,
from the state $\tilde{t}_1\tilde{r}_0$, the probabilities of going to states
$\tilde{t}_1\tilde{r}_0$, $\tilde{t}_1\tilde{r}_1$ and $\tilde{t}_1\tilde{r}_2$ are all $\frac{1}{3}$. Although this is slightly different from Algorithm~\ref{algorithm:above-threshold},
we will prove using this model to avoid nondeterministic choices, whether Algorithm~\ref{algorithm:above-threshold}
satisfies $\epsilon-$differential privacy can still be verified. Before that,
we first state the consistency of the outputs executed in the algorithm and the observation sequences in the model.

\begin{lemma}\label{lemma1}
Assume that there are two neighboring databases, $\od_1$ and $\od_2$, along with queries $f_i$ and threshold $t$ given as input of Algorithm~\ref{algorithm:above-threshold} and
the output is $A_n=a_1a_2...a_n = \bot\bot...\top$ with $n\ge 1$. Then there is an initial distribution pair $d_1$ and $d_2$ and an one-to-one mapping observation sequence $o_k$
such that $(\frac{1}{3})^{2n+1} \Pr_a(A_n|\od_1) = \Pr_m(o_n|d_1)$ and $(\frac{1}{3})^{2n+1} \Pr_a(A_n|\od_2) = \Pr_m(o_n|d_2)$, where $\Pr_a$ denotes the probability of getting the outputs in Algorithm~\ref{algorithm:above-threshold} and $\Pr_m$ denotes the probability of getting the observation sequence in the hmm model in  Fig.~\ref{figure:hmm-above-threshold}.
\end{lemma}

\begin{proof}
We prove by induction on the length of the output $A_n = \bot\bot...\top$ with $n$ symbols.
Assume the query results of neighboring databases $\od_1$ and $\od_2$
are $i_1,i_2,...$ and $j_1,j_2,...$, with $|i_k - j_k|<= 1$ for any fixed $k$.

Base case: $n=1$. If $A_1=\top$, the algorithm halts after comparing the first perturbed query with the perturbed threshold.
Naturally, there's only one state $t_{t}r_{i_1}$ with probability $1$ and the others with probability $0$ in the distribution $d_1$,
and only one state $t_{t}r_{j_1}$ with probability $1$ and the others with $0$ in $d_2$. Starting from these initial distributions,
we first observe a $\blank$ and then transit to the distribution with states $\tilde{t}_{t'}r_{i_1}$  and $\tilde{t}_{t'}r_{j_1}$
having non-zero probabilities, where $t'$ can be all possible values of perturbed threshold. The only observation shared by theses
states is $i_1j_1$, with observing probability $\frac{1}{3}$. Then queries are then added by noise and we come to a new distribution of perturbed query results and threshold $\tilde{t}_{t'}\tilde{r}_{k'}$, where $k'$ can be all possible values of perturbed query result.
This time we only choose states where perturbed query results are higher than the perturbed threshold and all these states share an observation of $\top$ with observing probability $1$.
Thus, under the observation sequence $o_1= \blank, i_1j_1, \top$, we follow the steps of Algorithm~\ref{algorithm:above-threshold} to make transitions in the model
and considering the observation probabilities, we can directly get $\frac{1}{3} \Pr_a(A_1|\od_1) = \Pr_m(o_1|d_1)$ and $\frac{1}{3} \Pr_a(A_1|\od_2) = \Pr_m(o_1|d_2)$.
Note that if $A_1=\bot$, then under the observation sequence $o_1= \blank, i_1j_1, \bot$, we can still conclude in a similar way that  $\frac{1}{3} \Pr_a(A_1|\od_1) = \Pr_m(o_1|d_1)$ and $\frac{1}{3} \Pr_a(A_1|\od_2) = \Pr_m(o_1|d_2)$.

Induction step: Assume we have $(\frac{1}{3})^{2n+1} \Pr_a(A_n|\od_1) = \Pr_m(o_n|d_1)$ and $(\frac{1}{3})^{2n+1}$ $\Pr_a(A_n|\od_2)$ $=\Pr_m$$($$o_n$$|$$d_2)$
with $n$ symbols in $A_n = \bot\bot...\bot$.
If we observe $A_{n+1} = \bot\bot...\top$ with $n+1$ symbols in the algorithm, by induction hypothesis,
we can immediately conclude that with $o_n = \blank, i_1j_1, \bot, i_2j_2, \bot, ..., i_nj_n, \bot$,
$(\frac{1}{3})^{2n+1} \Pr_a(A_n|\od_1) = \Pr_m(o_n|d_1)$ and $(\frac{1}{3})^{2n+1} \Pr_a(A_n|\od_2) = \Pr_m(o_n|d_2)$.
Since the last symbol in $A_n$ is $\bot$, new query $f_{n+1}$ must be posed in the algorithm and the new query results $i_{n+1}$, $j_{n+1}$
are going to be perturbed and compared with the perturbed threshold. In the hmm model, after observing the $n$th $\bot$ in the current distribution,
transition to the new distribution occurs where states $\tilde{t}_{t'}r_{i_{n+1}}$  and $\tilde{t}_{t'}r_{j_{n+1}}$ having non-zero probabilities,
with transition probability $\frac{1}{3}$. And the common observation in theses states is $i_{n+1}j_{n+1}$ with observing probability $\frac{1}{3}$. Then queries are further perturbed and we only filter the states $\tilde{t}_{t'}\tilde{r}_{k'}$ where the perturbed query results are above the perturbed threshold, with the common observation $\top$. Since we follow the steps of the algorithm to make transitions and add two multipliers of $\frac{1}{3}$ , we can conclude that under the sequence $o_{n+1} = o_n, i_{n+1}j_{n+1}, \top$, $(\frac{1}{3})^{2(n+1)+1} \Pr_a(A_{n+1}|\od_1) = \Pr_m(o_{n+1}|d_1)$ and $(\frac{1}{3})^{2(n+1)+1} \Pr_a(A_{n+1}|\od_2) = \Pr_m(o_{n+1}|d_2)$.

Note that the above mapping process is one-to-one correspondence. Thus the proof is finished.
\end{proof}

Then we can prove the differential privacy results.
\begin{theorem}
  The model used in Fig.~\ref{figure:hmm-above-threshold} satisfies
  $\epsilon-$differential privacy, i.e, Algorithm~\ref{algorithm:pufferfish-check}
  returns ``unsat'' for all the feasible observation sequences of lengths $k$,
  if and only if Algorithm~\ref{algorithm:above-threshold}
  satisfies $\epsilon-$differential privacy.
\end{theorem}

\begin{proof}
  Feasible observation sequences of lengths $k$ mean that Algorithm~\ref{algorithm:pufferfish-check} only checks paths that could represent
  complete execution paths of the Algorithm~\ref{algorithm:above-threshold}. For instance $k$ can't be $1$, which only
  represents the initial distributions of two databases in Fig.~\ref{figure:hmm-above-threshold}.
  Moreover, the observation sequences can't contain any of symbols $\spadesuit$,
  $\heartsuit$, $\diamondsuit$, $\clubsuit$: paths contain these symbols don't represent practical executions in the Algorithm~\ref{algorithm:above-threshold}.
  Since one of the neighboring data distributions must have
  probability $0$ for the unique appearance of these symbols in the half part of Fig.~\ref{figure:hmm-above-threshold},
  to filter out these observation paths, one just needs to add constraints that the observation probabilities
  can't be strictly equal to $0$.

  ``If'' direction: If the algorithm satisfies $\epsilon-$differential privacy,
  the probabilities of observing any length of outputs $A = \bot\bot...\top$ are mathematically similar starting from neighboring databases, $\od_1$ and $\od_2$. That is,
  \begin{equation}\label{eq1}
     e^{-\epsilon} \Pr_a(A|\od_2) \le \Pr_a(A|\od_1)\le e^\epsilon \Pr_a(A|\od_2).
  \end{equation}
  Using Lemma~\ref{lemma1}, we can directly conclude that
  \begin{equation}\label{eq2}
     e^{-\epsilon} \Pr_m(o|d_2) \le \Pr_m(o|d_1)\le e^\epsilon \Pr_m(o|d_2).
  \end{equation}
  Here $d_1$, $d_2$ and $o$ correspond to $\od_1$, $\od_2$ and $A$ in Lemma~\ref{lemma1}. Since $A$ can be any possible output,
  each can be mapped into an observation sequence , which makes up all the feasible observation sequences in the model.
  This actually verifies that our model satisfies $\epsilon$-differential privacy.

  ``Only if'' direction: If the algorithm doesn't satisfy $\epsilon-$differential privacy, there's a sequence $A = \bot\bot...\top$ with observing probabilities differing too much from initial distribution $\od_1$ and $\od_2$. By applying the similar analysis procedure, we can prove that the model in Fig.~\ref{figure:hmm-above-threshold} doesn't satisfy $\epsilon-$differential privacy.
\end{proof}

\bibliographystyle{splncs04}
\bibliography{refs}

\end{document}